\documentclass{m2an}
\usepackage[T1]{fontenc}
\usepackage[utf8]{inputenc}
\usepackage{hyperref}
\usepackage{amsmath,amsfonts,amssymb,amsthm,bbold,bbm,multirow,adjustbox}
\usepackage{pgf}
\usepackage{comment}

\newtheorem{theorem}{Theorem}[section]

\newtheorem{lemma}{Lemma}[section]
\newtheorem{definition}{Definition}[section]

\usepackage{tikz,pgfplots,comment}
\usepackage{nicefrac}
\usepackage{subfig}
\usepackage{booktabs}

\usepackage[english]{babel}%
\usepackage{csquotes}%

\def\tilde{\widetilde}
\def\b{\mathbf}
\DeclareMathOperator{\score}{score}

\usepackage{enumitem}
\usepackage[normalem]{ulem}
\begin{document}
\title{Nonlocal PageRank}\thanks{The work of F.D. has been supported by the GNCS--INdAM project ``Nonlocal models for the analysis of complex networks''}%
\author{Stefano Cipolla}\address{Department of Mathematics ``Tullio Levi-Civita'', University of Padua, 35121, Padova, Italy}
\author{Fabio Durastante}\address{Consiglio Nazionale delle Ricerche -- Istituto per le Applicazioni del Calcolo ``M. Picone'', Naples, Italy}
\author{Francesco Tudisco}\address{GSSI Gran Sasso Science Institute, 67100, L'Aquila, Italy}
\date{DATE}
\begin{abstract}
In this work we introduce and study a nonlocal version of the PageRank. In our approach, the random walker explores the graph using longer excursions than just moving between neighboring nodes. As a result, the corresponding ranking of the nodes, which takes into account a \textit{long-range interaction} between them, does not exhibit concentration phenomena typical of spectral rankings which take into account just local interactions. 
We show that the  predictive value of the rankings obtained using our proposals is considerably improved on different real world problems. 
\end{abstract}
\maketitle
\section{Introduction}
Identifying and quantifying important components in a dataset or a complex system modeled by a network, using only the topological structure of nodes and edges,  is a very important issue in exploratory data analysis.

Various specific tasks have been designed around this quite general problem, at a global scale---with the aim of providing insightful summary statistics such as clustering coefficients, robustness or total communicability \cite{benzi2013total,cohen2010complex,girvan2002community}---at an intermediate (or meso) scale---by identifying structures such as communities, anti-communities or core-periphery \cite{girvan2002community,fasino2017modularity,rombach2014core,tudisco2019nonlinear}---and at a local level---where we aim at quantifying various node or edge  properties such as triadic closure or edge communicability \cite{cohen2010complex,estrada2012structure}. 
Here we focus on the so called \textit{centrality} problem, where we aim at assigning an importance score to each node in the network in order to discover the  most relevant nodes. While this task aims at unveiling network features that take place at a very local scale (nodes), it is nowadays apparent that complex networks feature an intrinsic higher-order organization~\cite{benson2016higher} and that  centrality scores should exploit the structure of the network as a whole in order to unveil insightful network properties that are otherwise overlooked. To this end, in this work we propose a simple generalization of the renowned PageRank centrality which forces the global network structure into this local scale centrality model. 

PageRank  had its fortune due to its employment in early versions of the Google search engine \cite{page1999pagerank}. The main idea of this centrality model is a mutually reinforcing definition of importance: the importance of a node is influenced by the importances of the nodes it connects to. Equivalently,  PageRank
centrality can be be interpreted as the average amount of time that a random walker
spends on each node as the length of the walks tend to infinity. While this recursive definition clearly involves the global structure of connections in the network, at each time step the random walker moves from one node to another taking into account only the direct neighbors of that node. Thus, while this definition implies that each node score is influenced by all the other node importances, the classical PageRank centrality often results into a localized measure, as it typically happens for typical eigenvector--based centrality scores~\cite{PhysRevE.90.052808}. 

Different strategies have been considered in recent literature to overcome this issue. On the one end, higher-order adjacency tensors have been employed to model higher-order neighborhoods made by hyperedges containing three or more nodes. This approach is typically characterized  by the use of hypergraphs or simplicial complexes~\cite{benson2019three,arrigo2019framework,cipolla2019shifted,MR3771541}. On the other hand, non Markovian stochastic processes with memory have been used to model random walks that take into account longer paths of connections \cite{benson2017spacey,fasino2019higher,arrigo2017non,cipolla2019extrapolation}.

Following this second line of research, in this work we  propose a nonlocal version of the classical PageRank model based on the usage of the L{\'e}vy random walk, i.e. the usage of an \emph{anomalous} nonlocal diffusion that employs  a one-parameter family of decaying transition probabilities. In the case of undirected networks, this type of random walk was considered  for example in~\cite{riascos2012long}.  The main idea of our approach is to move from the original exploration strategy of the network exploited in the PageRank, where the random walker moves between neighboring nodes with uniform probability, to a strategy that permits longer excursions between the nodes of the network, i.e., it allows us to move from a node $i$ to any other node that is connected to $i$ through a path of any length.  %
These types of longer length interactions enhance the navigability of the network and thus allow for a faster exploration,    as observed in other related contexts, including fractal small-worlds networks~\cite{PhysRevE.74.017101}, lattices~\cite{PhysRevLett.104.018701} and general multi-hopper models on digraphs \cite{Weng2015,estrada2017random}. 
In particular, very related to our approach is the concept of path Laplacians introduced by Estrada et al.  in \cite{estrada2012path} and further analyzed in \cite{estrada2017path,estrada2018path}.%

The reminder of the paper is structured as follows:
We start  by fixing the notation and by recalling the standard PageRank model in Section~\ref{sec:thepagerankalgorithm}. Then, in Section~\ref{sec:nonlocalpagerankalgorithm} we introduce the proposed nonlocal PageRank model, which is based on the choice of a distance function between the nodes and a one-parameter family of  L{\'e}vy--type random walks on the graph. In Section~\ref{sec:shortest_path_PageRank} we perform an asymptotic analysis on the selection of the parameter that defines the transition probabilities and we show that the classical PageRank follows as a special case of the new model for large values of the parameter and  when the chosen distance is the standard shortest-path distance. However, different values of the parameter allow us to obtain models that are more stable and less localized, as discussed in Section \ref{sec:stability_non_locality}. These properties help improving a number of network mining tasks: In Section \ref{sec:link_prediction} we show how different values of the parameter affect the behavior of the model in the context of link prediction. Whereas in Section \ref{sec:choosing_the_distance} we discuss how different choices of the distance function can affect the model. In particular, choosing suitable---and possibly problem related---distance functions gives us an additional flexibility  which allows us to improve the quality of the resulting centrality assignment. We highlight this by considering the London underground train test problem, where we design a ``metro distance'' that takes into account the multilayer structure of the network (given by the several underground train lines) and that compares favorably with other PageRank--like centralities. 

\subsubsection*{Data and software}
All data and software used in this work are available online. For convenience, we list below all the network data used in the different sections of the paper. We detail additional information on the various datasets when appropriate in the text.
\begin{itemize}[noitemsep]
	\item {\tt USAir97}, this is a directed network of air traffic in the U.S.A.\, available from {\tt Pajek} repository~\cite{nr}.
	\item {{\tt EUair}, this is a thirty-seven layer network, each one corresponding to a different airline operating in Europe~\cite{Cardillo2013}. We use the aggregate graph, consisting of the union of all the layers.}
	\item {\tt Barcelona}, this is the directed transportation network for the city of Barcelona (Spain) from the Research Core Team collection~\cite{transportationnetwork}.
	\item \texttt{adjnoun}, this undirected network represents common occurrences of adjectives and nouns in the ``David Copperfield'' novel by C. Dickens~\cite{MR2282139}.
	\item \texttt{zachary}, this is a (small) social network of a university karate club~\cite{MR2282139}.
	\item {\tt gre\_115}, this directed network is available from the Harwell-Boeing  collection~\cite{MR3363405}.
	\item {{\tt cage9}, this is a network obtained from a DNA electrophoresis model~\cite{VANHEUKELUM2002313} in which a polymer is modeled as a chain of ``monomers'' connected by bonds.}
	\item {{\tt delaunay\_n10}/{\tt delaunay\_n12}, are networks obtained from the adjacency matrices relative to two Delaunay triangulations of random points in the unit square~\cite{5470485}. }
	\item {{\tt 3elt}, is an example graph from the AG-Monien Graph Collection by Ralf Diekmann and Robert Preis (see~\cite{MR2865011}).} 
	\item {\tt tube} is the undirected multilayer network of London's underground trains connections.
\end{itemize}
The {\tt tube} network was created by us  starting from the dataset developed in \cite{de2014navigability} and it is made of 13 layers: one layer for Docklands Light Railway (DLR) trains, one layer for overground trains and eleven underground train layers, one for each line. The dataset includes also geographical coordinates of the nodes and passengers usage statistics across several years (2008--2017), obtained from ORR London Datastore \cite{london_usage}. Both this dataset and the software we developed for the experiments shown in the paper are available at 	\url{https://github.com/Cirdans-Home/NonLocalPageRank}

\section{The PageRank algorithm}
\label{sec:thepagerankalgorithm}

A \emph{digraph}, or directed graph,  $\Gamma = (V,E)$ is defined by a set of $n$ nodes $V=\{v_1,\ldots,v_n\} \equiv \{1,\ldots,n\}$, and a set of ordered edges $E = \{(i,j)\ :\ i,j \in V \} \subseteq V \times V$ representing the connections between the nodes. A \emph{walk} of length $k$ in $\Gamma$ is a list of nodes $i_1,\ldots,i_k,i_{k+1}$ such that $(i_j,i_{j+1}) \in E$, $\forall j=1,\ldots,k$. If the first and the last edge coincides then the walk is called a closed walk. If no repeated nodes appear in the sequence then the walk is called a \emph{path}, while a path in which only the first and last node coincide is called a \emph{cycle}. We consider here only \emph{loop-less} graphs, i.e.,  edges of the form $(i,i) \in E$ are not allowed. We say that a digraph is \emph{strongly connected} if there exists a path between every pair of nodes.

Every graph $\Gamma$ can be represented as a binary adjacency matrix $A = (a_{i,j})$ with
\begin{equation*}
a_{i,j} = \left\lbrace\begin{array}{ll}
1 & (i,j) \in E,  \\
0 & \text{ otherwise}. 
\end{array}{}\right.
\end{equation*}{}
As we allow directed graphs, such matrix will not be symmetric, in general. 

Let $\boldsymbol{1}$ be the vector of all ones of size $n$, and let $A$ be the adjacency matrix of the digraph $\Gamma = (V,E)$ with $|V| = n$.  We  define the diagonal matrix $D$ of the out-degrees of $\Gamma$ as
\begin{equation}\label{eq:diagonalmatrix_of_out_degrees}
D_{\text{out}} = \operatorname{diag}(A \boldsymbol{1}) = \operatorname{diag}(d_1^{\text{out}},\ldots,d_n^{\text{out}}),
\end{equation}
in which each diagonal entry $d_i^{\text{out}}$ represents the number of outgoing edges from the node $i$.

We can now build a random walk on $\Gamma$ by considering the following transition matrix $P = (p_{i,j})$
\begin{equation}\label{eq:local_transition_matrix}
p_{i,j} = \left\lbrace \begin{array}{ll}
\nicefrac{1}{d_i^{\text{out}}}  & (i,j) \in E,  \\
0  & \text{ otherwise}. 
\end{array}{} \right.
\end{equation}
Note that $P$ is tightly related to the adjacency matrix and the diagonal matrix of the out degrees. In fact, a compact form for \eqref{eq:local_transition_matrix} reads $P = D_{\text{out}}^{-1}A$, where the matrix $D_{\text{out}}^{-1}$ is defined by setting the inverse of zero diagonal entries to zero by convention.  

A random walker that obeys the transition matrix $P$ has  equal chance of moving from a node to any of its out--neighbors. Note that by following this transition we could end in a \textit{cul--de--sac}, represented by a node $i$ with $d_i^{\text{out}} = 0$. To avoid this circumstance we modify $P$ in~\eqref{eq:local_transition_matrix} to permit the walker to teleport to any other location in the graph with some probability, i.e., we define the \emph{PageRank} transition matrix
\begin{equation}\label{eq:page_rank_matrix}
G = c \tilde{P} + \frac{1-c}{n} \boldsymbol{1}\boldsymbol{1}^T, \quad c \in (0,1],
\end{equation}
where $\tilde{P}$ is the matrix $P$ in which each zero row has been replaced by the uniform vector $\nicefrac{\boldsymbol{1}^T}{n}$. The matrix $G$ is a positive, row stochastic matrix. Thus, by the Perron-Frobenius Theorem, it admits  a unique positive and dominant left eigenvector $\mathbf{s}$ corresponding to the eigenvalue $1$ (see e.g.\ \cite{MR1298430}).  %
In other words, the Markov chain with transition matrix $G$ has a unique and positive stationary distribution \begin{equation*}
\mathbf{s}^T  = \mathbf{s}^TG,  \quad \b s^T \b 1 =1, \quad s_i>0, \forall\ i=1,\ldots,n.
\end{equation*}
Such $\mathbf{s}$ is called the \emph{PageRank} vector of $\Gamma$, and its $i$th entry provides a measure of the importance of node $i$ in the digraph $\Gamma$. Figure~\ref{fig:classicalpagerank} illustrates an example PageRank vector on the Zachary's karate club network.
\begin{figure}[t]
	\centering
	\includegraphics[width=0.85\columnwidth]{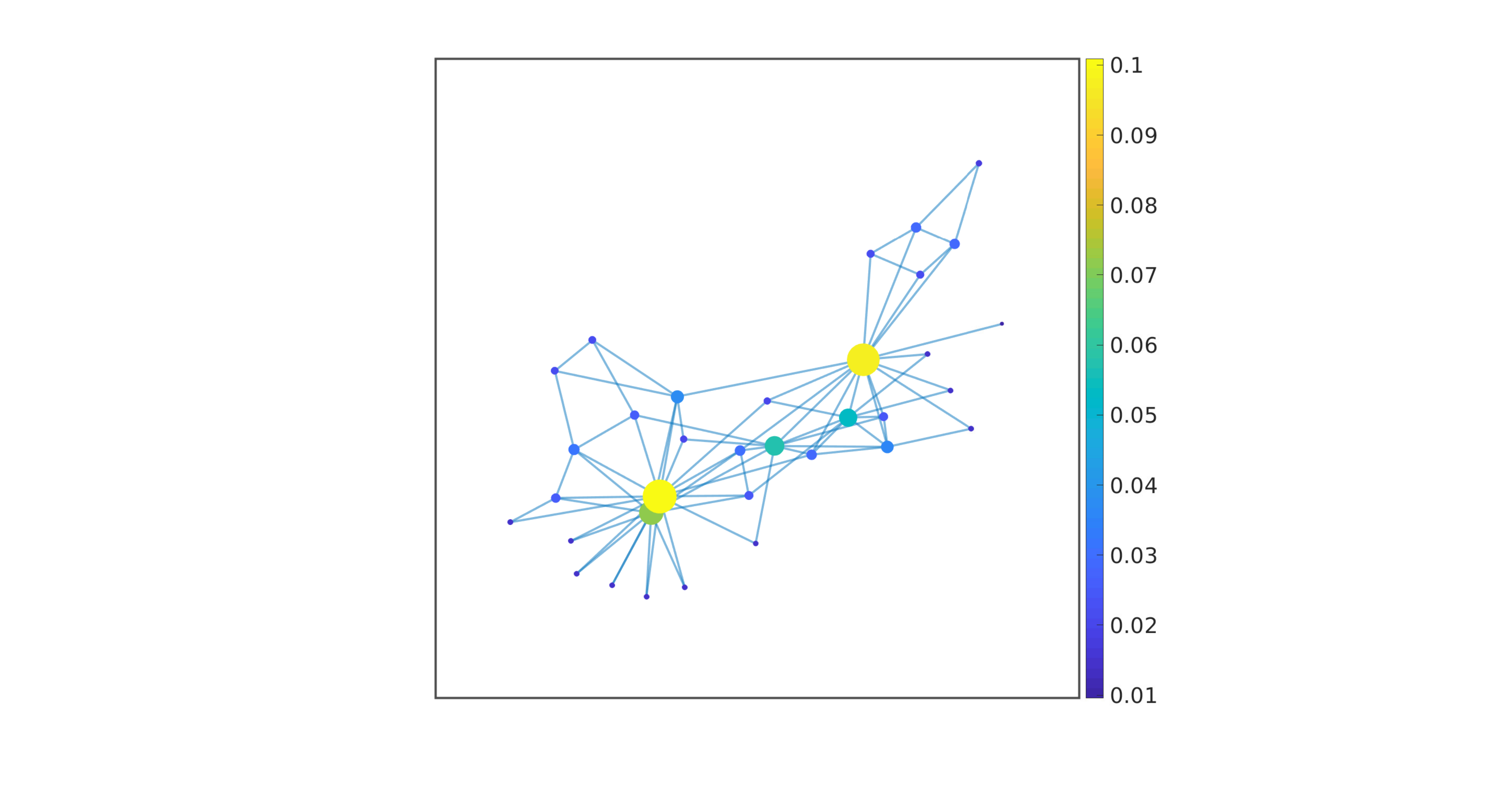}
	\caption{Pictorial representation of the PageRank vector $\b s$ for the Zachary's karate club digraph~\cite{girvan2002community}, computed with $c = 0.85$.}
	\label{fig:classicalpagerank}
\end{figure}

From the modeling point of view, the teleportation factor $c$ included in~\eqref{eq:page_rank_matrix} can be interpreted in terms of a ``surfer'' navigating the digraph $\Gamma$ of the web hyperlinks. When moving from a page to another, the surfer may decide to follow one of the hyperlinks that are listed in the web page, choosing  among them with uniform probability. The surfer makes this choice with probability $c$. Whilst, with probability $1-c$, they may decide to ``teleport'' onto a new  web page, in principle not connected with the current one, choosing again uniformly at random. This is an example of a nonlocal behavior in which it is possible to end up being in nodes that are far away from the original starting node. However, we have only partial control on this longer jumps. In fact, we are only allowed to either tune the parameter $c$ or to  introduce a ``personalized version of the PageRank'' where the surfer chooses to teleport onto a new page $j$ with a probability $v_j$ that depends on the destination page $j$, rather than choosing among all the possible web pages uniformly at random. From a mathematical viewpoint, this second choice means we replace the PageRank matrix \eqref{eq:page_rank_matrix} with $c\tilde P + (1-c)\b 1 \b v^T$, being $\b v = (v_1, \dots, v_n)$.

In the next Section~\ref{sec:nonlocalpagerankalgorithm} we introduce a modification of this model that allows us to include one additional level of nonlocality to the original PageRank centrality.

\section{The nonlocal PageRank model}
\label{sec:nonlocalpagerankalgorithm}

Let us suppose that we are  at a node $i$ in the digraph $\Gamma$. Following from the discussion of Section~\ref{sec:thepagerankalgorithm} we know that we have now two possibilities: going to a node that is connected to the present one, or to teleport away in the graph, possibly giving preference to certain nodes over others. However, this preference does not depend on the current node $i$ and, thus, does not take into account the fact that one is typically more inclined to teleport to some page that is somewhat related to $i$. 
To overcome this issue, we define here a process where  the probability to move to a node $j$,  that lies on a walk that contains $i$, is large when $j$ is  close  to $i$ and  decreases the further away we move from $i$. 

To this end, we let $\delta:V\times V\to \mathbb R_+$ be a distance function on $\Gamma$ and define the transition probability matrix $P_\alpha = (p_{i,j})$ as 
\begin{equation}\label{eq:levy_transition}
p_{i,j} = \begin{cases}
\frac{f_\alpha\big(\delta(i,j)\big)}{\sum_{k \neq i} f_\alpha\big(\delta (i,k)\big)} & \delta(i,j) < \infty\\[-.9em] & \\
0 & \text{ otherwise} 
\end{cases}\, ,
\end{equation}
where $f_{\alpha}(x): \mathbb{R}_+ \rightarrow \mathbb{R}_+$ is a family of nonnegative and nonincreasing functions parametrized by $\alpha \in \mathbb{R}_+$. The family $f_\alpha$  will be used to regulate the trade-off between the importance of the nodes at a short--range and the ones that are further away. 
Note that the distance $\delta$ does not need to define a metric on $\Gamma$ and, for example, does not need to be symmetric, i.e.{,} we allow $\delta(i,j)\neq \delta(j,i)$. In particular, observe that this is the situation if we choose $\delta$ to be the shortest path distance, as we will discuss in more details in the next section. Also note that, if $\Delta$ denotes the distance matrix such that $\Delta_{ij}=\delta(i,j)$, then it holds
\begin{equation*}
P_\alpha = \text{diag}(f_\alpha(\Delta)\b 1)^{-1}f_\alpha(\Delta)
\end{equation*}
where $f_\alpha$ is applied entrywise and the inverse of $+\infty$ is set to zero by convention.

As  for the standard PageRank case, a random walker obeying the transition rule \eqref{eq:levy_transition} could  end up in a node $i$ such that $\delta(i,j)=+\infty$ for all other $j\in V$ (i.e.\ the $i$-th row of $P_\alpha$ is all zero) and get stuck there. Thus we modify the state transitions to include a teleportation, i.e., we let
\begin{equation}\label{eq:nonlocalpagerankmatrix}
G_\alpha = c \tilde{P}_\alpha + \frac{1-c}{n}\boldsymbol{1}\boldsymbol{1}^T, \quad c \in (0,1], 
\end{equation}
where $\tilde{P}_\alpha$ is the matrix $P_\alpha$ in which each zero row has been replaced by $\nicefrac{\boldsymbol{1}^T}{n}$. The matrix $G_\alpha$ is again a positive row stochastic matrix and, by the Perron-Frobenius Theorem, there exists a unique $\b s_\alpha$ with positive entries, such that $\mathbf{s}_{\alpha}^T = \mathbf{s}_{\alpha}^TG_\alpha$, $\b s_\alpha^T \b 1 =1$. We call $\b s_\alpha$ the \textit{nonlocal PageRank} vector.

Clearly, the nonloncal PageRank depends on the choice of the distance $\delta$ and on the choice of the family of decaying functions $f_\alpha$. In the next section we show that, when $\delta$ is the shortest path distance and $f_\alpha$ is suitably defined, the nonlocal PageRank interpolates the standard PageRank and the parameter $\alpha$ can be used to tune the ``amount of nonlocality'' we want to take into account. 

\section{Shortest path distance and convergence to the PageRank}
\label{sec:shortest_path_PageRank}

As we will see in Section~\ref{sec:choosing_the_distance}, different choices of the distance $\delta$ can be used, depending on the application. In fact, this flexibility is one of the key advantages of our proposed framework, which allows us to model nonstandard node-node interactions and thus to capture node properties that are overlooked otherwise. On the other hand, the arguably most  natural choice for distance $\delta$ is the \emph{shortest path} distance, whose definition we recall below
\begin{definition}\label{def:shortestpath_distance}
	Given a digraph $\Gamma = (V,E)$  the \emph{shortest-path} distance $\delta_\Gamma(i,j)$ between any two nodes $i,j\in V$ is the smallest length of any  path from $i$ to $j$. If there exists no such a path, we let $\delta_\Gamma(i,j) = +\infty$.
\end{definition}

This choice of the distance function allows us to retrieve the classical PageRank as a special case of our nonlocal PageRank model, provided the family $f_\alpha$ satisfies the following nonrestrictive decaying condition. 
\begin{definition}
	We say that $f_\alpha$ is a smoothing family of functions for the distance $\delta$ if $f_\alpha(x)\to 0$ and  $f_\alpha(x)/f_\alpha(1)\to 0$ as $\alpha\to\infty$, for all $x$ in the set
	$$
	\Omega_\delta = \{x: x\geq \delta(i,j), \text{ for all } (i,j)\notin E\} \subseteq \mathbb R_+\, .
	$$
\end{definition}
When $\delta =\delta_\Gamma$ we have $\Omega_{\delta_\Gamma}=[2,\infty]$, and we can easily produce examples of smoothing families by considering any nonincreasing function $f$ such that $f(2) <1$ and then defining $f_\alpha(x) = f(x)^\alpha$. For example, we can choose, as in~\cite{riascos2012long}, the functions
\begin{equation*}
f_\alpha(x) = \frac{1}{x^{\alpha}}\, .
\end{equation*}
In this case, the resulting transition matrix  $P_\alpha$ describes a L{\'e}vy random walk on $\Gamma$. Another example is the exponential function
\begin{equation*}
f_\alpha(x) = e^{-\alpha\, x}\, .
\end{equation*}

For any smoothing family of functions we have
\begin{lemma}\label{lem:smoothing_family}
	Let $f_\alpha$ be a smoothing family of nondecreasing functions for the distance $\delta = \delta_\Gamma$. Then $P_\alpha \to P$ entrywise as $\alpha\to \infty$, where $P$ is the PageRank transition matrix \eqref{eq:page_rank_matrix}.
\end{lemma}
\begin{proof}
	Note that $\sum_k f_\alpha(\delta_\Gamma(i,k)) = d^{\text{out}}_i f_\alpha(1) + \sum_{k:\delta_\Gamma(i,k)\geq 2} f_\alpha(\delta_\Gamma(i,k))$. Thus, simple algebraic manipulations show that $P_\alpha = P + Y_\alpha$, with
	\begin{equation*}
	(Y_\alpha)_{ij} = \begin{cases}
	\displaystyle -\frac{1}{d^{\text{out}}_i} \frac{\sum_{k:\delta_\Gamma(i,k)\geq 2} f_\alpha(\delta_\Gamma(i,k))}{d^{\text{out}}_i f_\alpha(1) + \sum_{k:\delta_\Gamma(i,k)\geq 2} f_\alpha(\delta_\Gamma(i,k))}& (i,j)\in E\\
	\displaystyle \frac{f_\alpha(\delta_\Gamma(i,j))}{\sum_{k \neq i} f_\alpha(\delta_\Gamma(i,k))} & 2\leq \delta_\Gamma(i,j)<\infty \\
	0 & \delta_\Gamma(i,j) = \infty
	\end{cases} \, .
	\end{equation*}
	As $f_\alpha(x)\to 0$ for $\alpha \to \infty$ and $x\geq 2$, we have $(Y_\alpha)_{ij}\to 0$, and the proof is complete.
\end{proof}
\begin{figure}[t]
	\centering
	\includegraphics[width=\columnwidth]{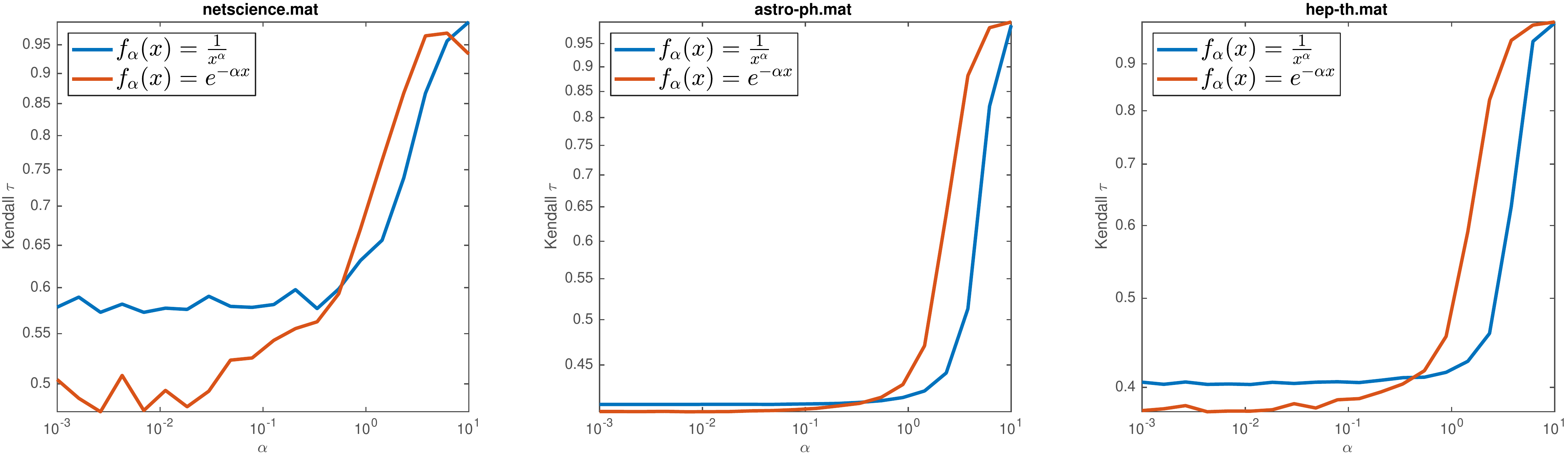}
	\caption{ {The figure shows Kendall's $\tau$ correlation coefficient of PageRank vs Nonlocal PageRank for different  $f_\alpha(x)$ and different values of $\alpha$ for the shortest path distance. Problems sizes (left to right) $n=1589$, $n=16706$ and $n=8361$.}}
	\label{fig:shortestpath}
\end{figure}

Note that, when $f_\alpha(x)=f(x)^\alpha$ is defined via a power-law (and $f(x)$ is bounded) we obviously have $f_\alpha(x)\to 1$ as $\alpha\to 0$, implying that $P_\alpha$ converges to the uniform transition matrix where $i$ transitions to every node $j$ at finite distance from $i$ with equal probability. This observation, combined with the Lemma above, shows that choosing values of $\alpha$ reasonably far from $0$ and $\infty$ allows us to define a model that interpolates between a standard Markov chain on the graph and a purely uniformly random model.  
This is also shown by Figure~\ref{fig:shortestpath}, where we 
{plot the Kendall's $\tau$ correlation coefficient  between the nonlocal PageRank centrality score---obtained with different values of $\alpha$---and the standard PageRank centrality on the datasets {\tt netscience, astro-ph} and {\tt hep-th}.  In particular, the figure highlights the result in Lemma~\ref{lem:smoothing_family} showing that as $\alpha \rightarrow +\infty$ the correlation between the nonlocal PageRank and the PageRank converges to $1$, thus showing the convergence to the standard PageRank algorithm.}
Clearly, while the two extreme choices $\alpha\to 0$ and $\alpha\to\infty$ do not require to compute $\delta_\Gamma$, one of the drawbacks of using the nonlocal PageRank transition matrix $P_\alpha$ for other choices of $\alpha$ is the computation of the shortest--path distance matrix from Definition~\ref{def:shortestpath_distance}. The algorithm of choice for computing all the distances between the nodes is either the Floyd--Warshall algorithm~\cite{Floyd:1962:A9S:367766.368168,Warshall:1962:TBM:321105.321107} or the Johnson's algorithm~\cite{Johnson:1977:EAS:321992.321993}. The first one works on weighted digraphs with positive or negative edge weights, and no negative cycles, and it has a running time of $O(n^3)$. The second one allows only for few negative edge weights and has a running time of $O(n^2\log(n)+nm)$ where $m = |E|$. As it is clear from the running times estimates, the Floyd-Warshall algorithm is best suited for dense networks ($m \approx n$), while the Johnson's algorithm should be preferred in the case of sparse networks ($m \ll n$). The storage of $O(n^2)$ machine numbers has to be expected.

In the next section we show how different values of $\alpha$ can improve the stability and reduce the localization phenomenon of the standard  PageRank.

\subsection{Stability and nonlocality} \label{sec:stability_non_locality}

The classical PageRank centrality measure shares the same principal flaw of all the eigenvector centralities. The leading eigenvector
of the associated transition matrix can suffer the localization phenomenon~\cite{MR2302549,PhysRevE.90.052808,PhysRevE.99.012315}, i.e., most of the measure weight tends to concentrate around few most important nodes while giving to all the other nodes a small and numerically identical value, and thus ranking. This phenomenon has also the secondary effect of producing high variations in the value of the entries of the PageRank vector when the network topology faces a small perturbation, for example few edges are added or removed from the graph~\cite{MR2002174,MR3706916,pozza2017stability}. 

In a linear algebra terminology, this type of graph modification is equivalent to a structured perturbation of the transition matrix and is related to the concept of \textit{condition number of a Markov chain}~\cite{neumann-and-xu}. 

For a given $\alpha$, a given distance $\delta$ and a given smoothing function $f_\alpha$, we define a condition number of the nonlocal PageRank as a coefficient $\kappa_\alpha$ that quantifies the relative change in the nonlocal PageRank vector when changes occur in the graph. More precisely, suppose that $\overline \Gamma = (V,\overline E)$ is an edge-perturbation of $\Gamma=(V,E)$ and let $\overline G_\alpha$ be the nonlocal transition matrix \eqref{eq:nonlocalpagerankmatrix} of  $\overline \Gamma$. Then $\kappa_\alpha$ is a condition number of the nonlocal PageRank if
\begin{equation}\label{eq:def_general_cond_number}
    \frac{\|\mathbf{s}_\alpha-\overline{\mathbf{s}}_\alpha\|}{\|\mathbf{s}_\alpha\|} \leq \kappa_\alpha \, \frac{\|G_\alpha-\overline G_\alpha\|}{\|G_\alpha\|}
\end{equation}
and this bound holds for all possible perturbations of $\Gamma$. In~\eqref{eq:def_general_cond_number} $\mathbf{s}_\alpha$ and $\overline{\mathbf{s}}_\alpha$ are the nonlocal PageRank of $G_\alpha$ and $\overline G_\alpha$ respectively.  Note that property~\eqref{eq:def_general_cond_number} does not define a unique condition number $\kappa_\alpha$ and in fact several comparisons between different possible choices of $\kappa_\alpha$ have been studied~\cite{ChoandMeyer,Kirkland,neumann-and-xu}. 
Here we use a result of Seneta~\cite{MR932541} to characterize the stability of the nonlocal PageRank in terms of its norm-1 condition number. 

To this end, we use the \emph{ergodicity coefficient} of the  matrix sequence $\{G_\alpha\}_\alpha$. For a fixed $\alpha$, this coefficient is defined as  %
\begin{equation}\label{eq:ergodicity_coefficient}
\tau_1({G}_\alpha) = \sup_{\substack{\|\boldsymbol{\delta}\|_1 = 1 \\ \boldsymbol{\delta}^T \boldsymbol{1} = 0}} \|\boldsymbol{\delta}^T {G}_\alpha \|_1 = \frac{1}{2}\max_{j} \| G_\alpha^T (I - \mathbf{e}_j \boldsymbol{1}^T)\|_1 = \frac{1}{2} \max_{i,j}\sum_{k}\left| (G_\alpha)_{ik}-(G_\alpha)_{jk}\right|
\end{equation}
and the following quantity
$$
\operatorname{cond}_1(\mathbf{s}_\alpha) = \frac{1}{1-\tau_1(G_\alpha)}
$$
is a condition number for the nonlocal PageRank, in the sense that inequality~\eqref{eq:def_general_cond_number} holds for $\kappa_\alpha=\operatorname{cond}_1(\mathbf{s}_\alpha)$ and $\|\cdot\|=\|\cdot\|_1$, see e.g.\ \cite{MR932541}. Therefore, if we define the remainder $R =  G_\alpha - \overline{G}_\alpha$ and we express $G_\alpha$ as the structured perturbation $G_\alpha = \overline{G}_\alpha + R$, with $R$ such that $R \boldsymbol{1} = 0$, we can recover the following perturbation bound on the associated PageRank vectors
\begin{equation}\label{eq:convergence_bound}
\| \mathbf{s}_\alpha - \overline{\mathbf{s}}_\alpha \|_1 \leq  \|R\|_1 (1 - \tau_1(G_\alpha))^{-1}  \, .
\end{equation} 
In particular, as $\|\mathbf{s}_\alpha^T\|_1 = \|G_\alpha\|_1 = 1$ and  $|\lambda(G_\alpha)|\leq \tau_1(G_\alpha)$, for any eigenvalue $\lambda(G_\alpha)$ of $G_\alpha$ such that $\lambda(G_\alpha)\neq 1$ (see e.g.\ \cite{tudisco2015note}), the relation~\eqref{eq:convergence_bound} gives us a bound on the condition number of $\mathbf{s}_\alpha$ for relative changes in $G_\alpha$ in the norm sense, i.e.,
\begin{equation}
\operatorname{cond}_1(\mathbf{s}_\alpha) \leq \frac{1}{1 - |\lambda(G_\alpha)|}, \quad \forall \, \lambda(G_\alpha) \neq 1.
\end{equation}

It is then interesting to evaluate the behavior of such quantity with respect to the parameter $\alpha$, and indeed it is possible to prove the following result.
\def\hat{\widehat}
\begin{theorem}\label{pro:conditioning} Let $f_\alpha$ be a family of decreasing smoothing functions for the distance $\delta=\delta_\Gamma$.
Then, there exists $\alpha_0$ such that $\tau_1(G_{\alpha}) \leq \tau_1(G)$ for all $\alpha \geq \alpha_0$. %
\end{theorem}

\begin{proof}
From~\eqref{eq:ergodicity_coefficient} it is clear that the thesis holds if and only if $\tau_1(\tilde P_\alpha)\leq \tau_1(\tilde P)$. Moreover, since $P$ and $P_\alpha$ have the same zero rows and $\tilde P$ and $\tilde P_\alpha$ coincide on those rows, again from~\eqref{eq:ergodicity_coefficient} we see that it is enough to show that $\tau_1(P_\alpha)\leq \tau_1(P)$, for all $\alpha\geq \alpha_0$.  

Recall that, by Lemma~\ref{lem:smoothing_family}, we have $\displaystyle \lim_{\alpha\to \infty}P_\alpha=P$. 
If $\delta_\Gamma(i,j)=1$ for all $i,j$ then it is easy to see that $P_{\alpha}=P$ for all $\alpha$ and thus the thesis is straightforward. 

Assume that there exist  $i$ and $j$ such that $\delta_\Gamma(i,j)>1$.
	Consider the matrix $\hat{P}_{\alpha}$ obtained modifying the $i$-th row of $P$ into the $i$-th row of $P_\alpha$, i.e. let
	\begin{equation*}
		(\hat{P}_{\alpha})_{i,:}=
		\begin{bmatrix}
		\frac{f_{\alpha}(\delta_{\Gamma}(i,1))}{\sum_{\ell\neq i}f_{\alpha}(\delta_{\Gamma}(i,\ell))} &  \cdots &  0 &  \cdots & \frac{f_{\alpha}(\delta_{\Gamma}(i,n))}{ \sum_{\ell\neq i}f_{\alpha}(\delta_{\Gamma}(i,\ell))}
		\end{bmatrix}
	\end{equation*}
	 where the zero entry is the diagonal $(\hat P_\alpha)_{i,i}$,  and let  $(\hat{P}_{\alpha})_{j,:}=P_{j,:}$ for all $j \neq i$. Observe that if $k \in E_i:=\{k \hbox{ s.t. } (i,k) \in E \}$, it holds
	 
	 \begin{equation} \label{eq:pattern_ineq}
	 (\hat{P}_{\alpha})_{i,k}= \frac{1}{\displaystyle d_i+ \sum_{\ell : \delta_{\Gamma}(i,\ell) \geq 2} f_{\alpha}(\delta_\Gamma(i,\ell))/f_{\alpha}(1)} <  \frac{1}{d_i}=P_{i,k},
	 \end{equation}
	 were the strict inequality holds from our assumption. More precisely,
	 \begin{equation}\label{eq:difference}
	 	(\hat{P}_{\alpha})_{i,k}=P_{i,k}+D^{\alpha}_i
	 \end{equation}
	 where $D^{\alpha}_i:= 1/(d_i+ \sum_{\ell : \delta_{\Gamma}(i,\ell) \geq 2}f_{\alpha}(\delta_\Gamma(i,\ell))/f_{\alpha}(1))-1<0$ for all $\alpha>0$.
	 Let us define $i_0(\alpha), j_0(\alpha)$ such that

	 \begin{equation*}
	 	\tau(\hat{P}_{\alpha})=\frac{1}{2} \sum_{k} |(\hat{P}_{\alpha})_{i_0(\alpha),k}-(\hat{P}_{\alpha})_{j_0({\alpha}),k}|.
	 \end{equation*}
Since $\hat{P}_{\alpha}$ converges, there exists $\overline{\alpha}_0$ s.t. $i_0(\alpha)= i_0$ and $j_0(\alpha)= j_0$ for all $\alpha \geq \overline{\alpha}_0.$	Moreover, since $\hat{P}_{\alpha}$ differs from $P$ only in the $i$-th row, to prove $\tau_1(\hat{P}_{\alpha}) \leq \tau_1(P)$, it is enough to consider 	the case $i_0=i, \;j_0 \neq i $.

In the setting we are considering we have  $(\hat{P}_{\alpha})_{j_0,k}={P}_{j_0,k}$ for all $k$. So, let us define the following sets
\begin{align*}
	&\hat{\mathcal{P}}^{\alpha}_i:=\{k : (\hat{P}_{\alpha})_{i_0,k} \geq {P}_{j_0,k}\}, \qquad  \hat{\mathcal{P}}^{\alpha}_{j_0}:=\{k : (\hat{P}_{\alpha})_{i_0,k} <  {P}_{j_0,k}\},\\
	&{\mathcal{P}}_i:=\{k : {P}_{i_0,k} \geq {P}_{j_0,k}\}, \qquad  {\mathcal{P}}_{j_0}:=\{k :  {P}_{i_0,k} <  {P}_{j_0,k}\}.
\end{align*}

It is easy to check that	 

\begin{itemize}
	\item $\sum_{k \in \hat{\mathcal{P}}^{\alpha}_i} (\hat{P}_{\alpha})_{i_0,k}-P_{j_0,k}= \sum_{k \in \hat{\mathcal{P}}^{\alpha}_{j_0}} P_{j_0,k}- (\hat{P}_{\alpha})_{i_0,k}$ (using the fact that every row is stochastic);
	
	\item $\hat{\mathcal{P}}^{\alpha}_i  \cap E_i \subseteq  \mathcal{P}_i  \cap E_i $ (using \eqref{eq:pattern_ineq}).
\end{itemize}

From the above observations, we have

\begin{equation}\label{eq:many}
\begin{split}
\tau_1(\hat{P}_{\alpha})= & \sum_{k \in \hat{\mathcal{P}}^{\alpha}_i} = ((\hat{P}_{\alpha})_{i,k}-P_{j_0,k})  \\
= & \sum_{k \in \hat{\mathcal{P}}^{\alpha}_i\cap E_i} ((\hat{P}_{\alpha})_{i,k}-P_{j_0,k})+\sum_{k \in \hat{\mathcal{P}}^{\alpha}_i \cap E_i^C} ((\hat{P}_{\alpha})_{i,k}-P_{j_0,k} ) \\
= & \sum_{k \in \hat{\mathcal{P}}^{\alpha}_i\cap E_i} ({P}_{i,k}-P_{j_0,k})+\sum_{k \in \hat{\mathcal{P}}^{\alpha}_i\cap E_i}D_i^{\alpha}+\sum_{k \in \hat{\mathcal{P}}^{\alpha}_i \cap E_i^C} ((\hat{P}_{\alpha})_{i,k}-P_{j_0,k} ) \\
\leq & \sum_{k \in \hat{\mathcal{P}}^{\alpha}_i\cap E_i} ({P}_{i,k}-P_{j_0,k})+\sum_{k \in \hat{\mathcal{P}}^{\alpha}_i\cap E_i}D_i^{\alpha}+\sum_{k \in \hat{\mathcal{P}}^{\alpha}_i \cap E_i^C}  \frac{f_{\alpha}(\delta_\Gamma(i,k))/f_{\alpha}(1)}{\displaystyle d_i+ \sum_{\ell : \delta_{\Gamma}(i,\ell) \geq 2} f_{\alpha}(\delta_\Gamma(i,\ell))/f_{\alpha}(1)}  \\
\leq & \sum_{k \in {\mathcal{P}}_i\cap E_i} ({P}_{i,k}-P_{j_0,k})+\sum_{k \in \hat{\mathcal{P}}^{\alpha}_i\cap E_i}D_i^{\alpha}+\sum_{k \in \hat{\mathcal{P}}^{\alpha}_i \cap E_i^C}  \frac{f_{\alpha}(\delta_\Gamma(i,k))/f_{\alpha}(1)}{\displaystyle d_i+ \sum_{\ell : \delta_{\Gamma}(i,\ell) \geq 2} f_{\alpha}(\delta_\Gamma(i,\ell))/f_{\alpha}(1)} ,
\end{split}
\end{equation}
where in the first inequality we used \eqref{eq:difference} and the definition of $(\hat{P}_{\alpha})_{i,k}$ and, in the second one, the fact that $\hat{\mathcal{P}}^{\alpha}_i  \cap E_i \subseteq  \mathcal{P}_i  \cap E_i $.

Since $\sum_{k \in \hat{\mathcal{P}}^{\alpha}_i\cap E_i}D_i^\alpha < 0$ for all $\alpha>0$ and $\sum_{k \in \hat{\mathcal{P}}^{\alpha}_i \cap E_i^C} ( \frac{f_{\alpha}(\delta_\Gamma(i,k))/f_{\alpha}(1)}{d_i+ \sum_{\ell : \delta_{\Gamma}(i,\ell) \geq 2} f_{\alpha}(\delta_\Gamma(i,\ell))/f_{\alpha}(1)} ) \to 0$ as $\alpha \to \infty$ ($f_{\alpha}$ is a smoothing family), there exists $\alpha_0^{(1)} \geq \overline{\alpha}_0$ s.t. 
\begin{equation*}
    \sum_{k \in \hat{\mathcal{P}}^{\alpha}_i\cap E_i}D_i^\alpha +\sum_{k \in \hat{\mathcal{P}}^{\alpha}_i \cap E_i^C}  \frac{f_{\alpha}(\delta_\Gamma(i,k))/f_{\alpha}(1)}{ \displaystyle d_i+ \sum_{\ell : \delta_{\Gamma}(i,\ell) \geq 2} f_{\alpha}(\delta_\Gamma(i,\ell))/f_{\alpha}(1)}  <0   \hbox{ for all  } \alpha \geq \alpha_0^{(1)}.
\end{equation*}
Finally, observing that  $\sum_{k \in {\mathcal{P}}_i\cap E_i} ({P}_{i,k}-P_{j_0,k}) \leq \frac{1}{2}  \sum_{k} |{P}_{i,k}-P_{j_0,k}| \leq \frac{1}{2} \max_{s,m} \sum_{k}|{P}_{s,k}-P_{m,k}|= \tau_1(P)$, we obtain from \eqref{eq:many}, that $\tau_1(\hat{P}_{\alpha}) \leq \tau_1(P)$ for all $\alpha \geq \alpha_0^{(1)}$.

Let us now define $\hat{P}_{\alpha}^{(1)}=\hat{P}_{\alpha}$. Observe that the above reasoning is still valid using $\hat{P}_{\alpha}^{(1)}$ in place of $P$ and considering as $\hat{P}_{\alpha}$ the matrix obtained modifying one more row of the new $\hat{P}_{\alpha}^{(1)}$ into the same row of $P_\alpha$ (overall modifying exactly two rows of the original  $P$ into the corresponding rows of $P_\alpha$). We call the matrix obtained in this way $\hat{P}_\alpha^{(2)}$. In this case we have that  there exists $\alpha_0^{(2)}$ such that, for all $\alpha \geq \alpha_0^{(2)}$, $\tau_1({\hat{P}_\alpha}^{(2)})\leq \tau_1(\hat{P}_{\alpha}^{(1)})$.

Applying this argument $n$ times we obtain that there exists $\alpha_0^{(n)}=:\alpha_0$ such that, for all $\alpha\geq \alpha_0$, the following  chain of inequalities holds
\begin{equation*}
	\tau_1({{P}_\alpha})=\tau_1(\hat{P}_{\alpha}^{(n)}) \leq \cdots  \leq \tau_1(\hat{P}_{\alpha}^{(1)})\leq \tau_1({P}),
\end{equation*}
concluding the proof.
\end{proof}

\begin{figure}[htbp]
	\centering
\includegraphics[width=\columnwidth]{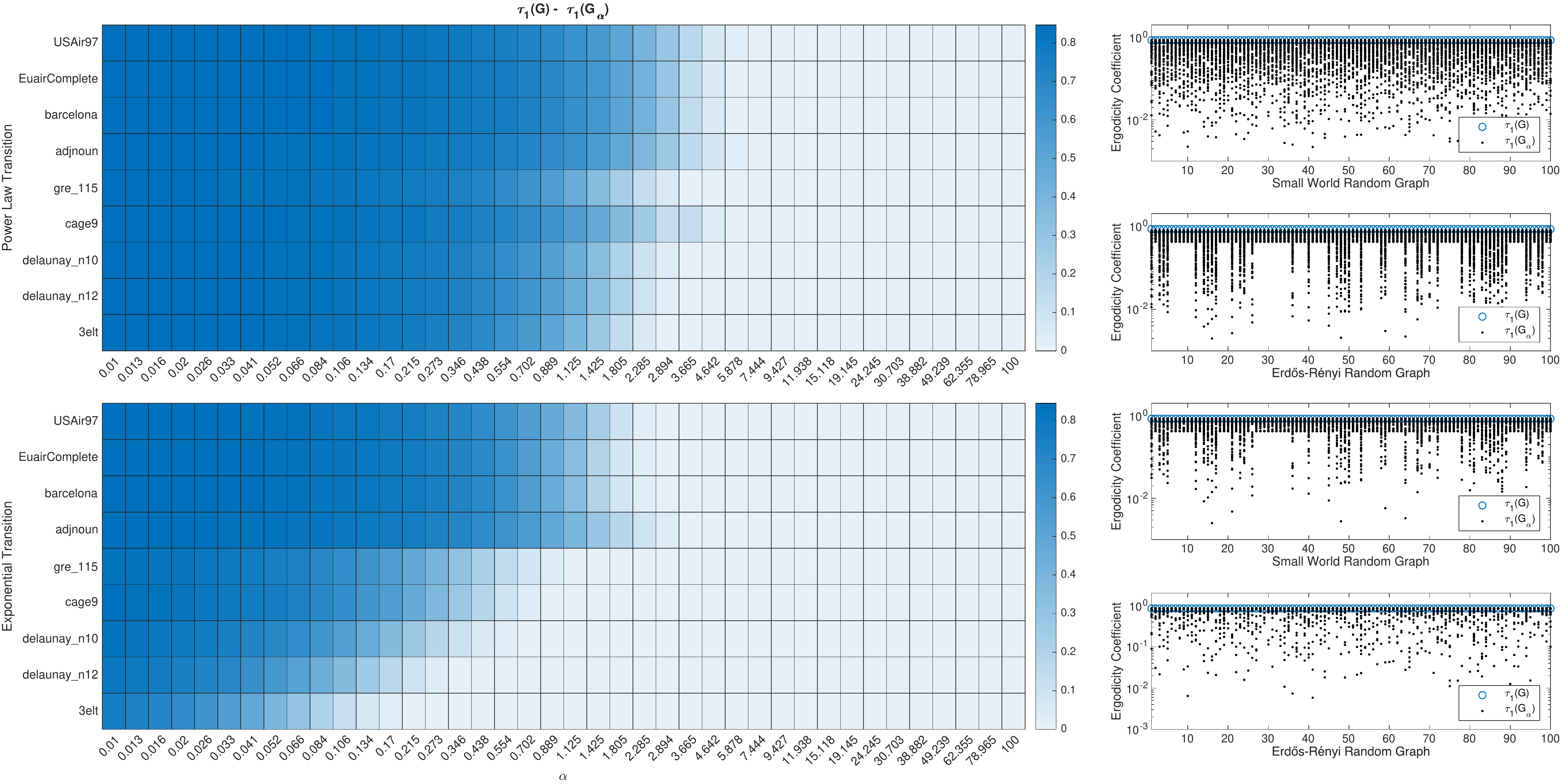}
\caption{{Behavior of the ergodicity coefficient $\tau_1$ for both the power law  $f_\alpha(x) = x^{-\alpha}$  and the exponential $f_\alpha(x) = \exp(-\alpha x)$ transitions. The heat-maps on the left-hand side represent the difference $\tau_1(G)-\tau_1(G_\alpha)$ for growing values of $\alpha$ on   a large set of test networks. Whereas, the panels on the right-hand side  report the behavior of the ergodic coefficients for one hundred random graphs of  Small World and Erd\H{o}s-R\'enyi type. For every graph the ergodic coefficient for the nonlocal PageRank is computed for one hundred different random alphas and, in every case, we observe that $\tau_1(G_\alpha) \leq \tau_1(G)$.}}
	\label{fig:ergodicity_bound}
\end{figure}
{The result above shows that when the nonlocality parameter $\alpha$ is large enough, the corresponding PageRank vector is more stable with respect to structured perturbations of $\Gamma$ or, equivalently, with respect to perturbations of the adjacency matrix. In fact, if $\mathbf s$ and $\mathbf s_\alpha$ are the PageRank and the nonlocal PageRank vectors, respectively, the inequality $\tau_1(G_\alpha)\leq \tau_1(G)$ implies $\operatorname{cond}_1(\mathbf s_\alpha)\leq \operatorname{cond}_1(\mathbf s)$.  On the other hand, extensive numerical experiments show that this inequality actually holds for all values of $\alpha$. We illustrate this result in Figure~\ref{fig:ergodicity_bound} where we compare the ergodicity coefficient $\tau_1(G_\alpha)$ of the nonlocal PageRank (for $\delta=\delta_\Gamma$)  with the coefficient $\tau_1(G)$ for the standard PageRank, for a broad range of values of the parameter $\alpha$ and for  both the example smoothing families of functions $f_\alpha(x)=x^{-\alpha}$ and $f_\alpha(x) = e^{-\alpha x}$. This figure shows that the nonlocal PageRank is actually more stable than its local counterpart for all values of $\alpha$.}  %

Clearly, the norm-1 bound we have obtained provides an indication on the entry-wise variation of the PageRank vector as well. To showcase this effect we consider the case of the undirected cycle graph $\mathcal{C}_n$, i.e., the graph on $n$ nodes containing a single cycle through all of them. It is easy to observe that for such graph the non--normalized PageRank vector  is $\mathbf{s}_\alpha = \boldsymbol{1}$, $\forall \alpha \geq 0$. Let us now add  a directed edge between the $\ell$th and the $1$st node of the graph, and evaluate how the PageRank vector for the relative value of $\alpha$ is modified. It is known that for $\alpha = \infty$ this modification produces a strong localization effect on the standard PageRank vector $\mathbf{s}_\infty$~\cite[Theorem~8.1]{MR3706916}. The bound obtained in Theorem~\ref{pro:conditioning} suggests that the relative change on the vector $\mathbf{s}_\alpha$ decreases with smaller values of $\alpha$, i.e.\ the localization effect is milder on the nonlocal PageRank vector. This is shown in  Figure~\ref{fig:cyclegraph} where we plot the PageRank vector of the cycle $\mathcal C_{100}$ perturbed with one additional edge that connects nodes 1 and 40,  for different values of $\alpha$ and for $\delta=\delta_\Gamma$.   Consistently with the analysis carried out in~\cite{MR3706916}, we observe two localized peaks that appear in the PageRank vector $\mathbf s_\infty$  entries corresponding to the nodes 1 and 40. At the same time, nonlocal versions of PageRank smooth out these peaks as soon as we let the value of $\alpha$ decrease.

\begin{figure}[t]
	\centering
	\includegraphics[width=\columnwidth]{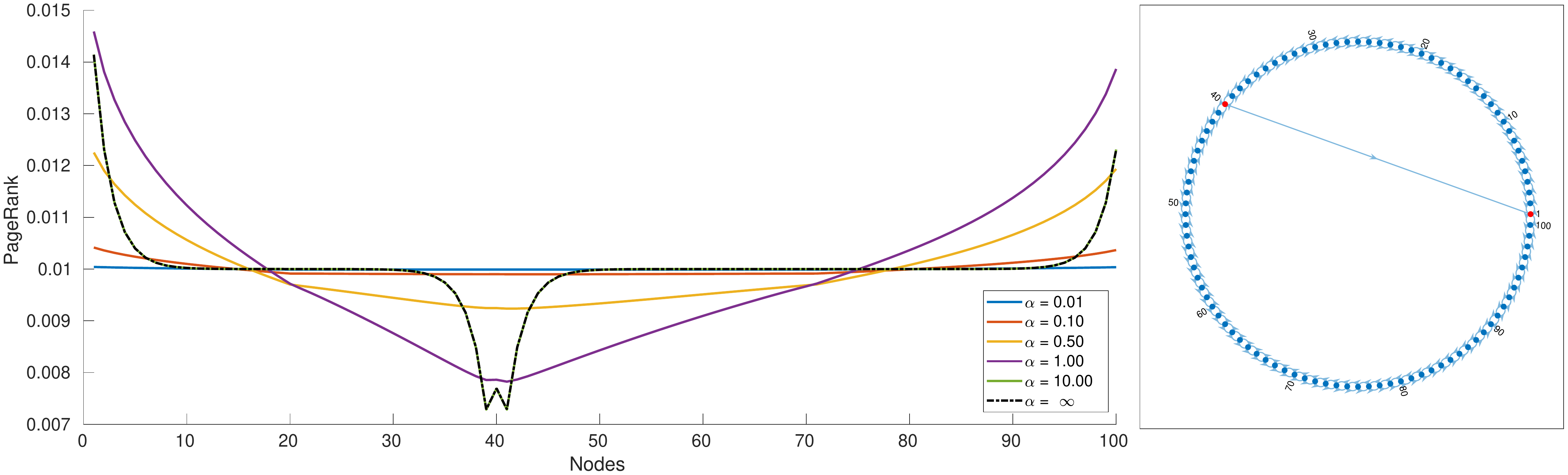}
	\caption{Locality of the PageRank vector for the undirected cycle graph with $n = 100$ nodes when a single directed edge between nodes $40$ and $1$ is added to the graph (right panel). The PageRank vector for the undirected cycle is the vector $\boldsymbol{1}$, (represented as the dashed black line). The smaller the value of $\alpha$ is, the more stable the measure is with respect to the addition of this new edge since the curves for the smaller values of $\alpha$ tends to the PageRank of the unmodified graph. On the other hand, for growing values of $\alpha$ the the curve representing the value of the nonlocal PageRank converges to curve of the standard PageRank (represented as the dot--dashed black line on the left panel).}
	\label{fig:cyclegraph}
\end{figure}
In Figure~\ref{fig:localization} we show the localization behavior of nonlocal PageRank vectors on the {\tt adjnoun} real-world network. As expected,  the standard PageRank algorithm  concentrates the measure around few nodes (hubs)  and assigns small and almost indiscernible values to the nodes that occupy the lower part of the ranking, whereas tuning the parameter $\alpha$ allows us spread the measure more evenly. 

\begin{figure}[htbp]
	\centering
	\includegraphics[width=0.75\columnwidth]{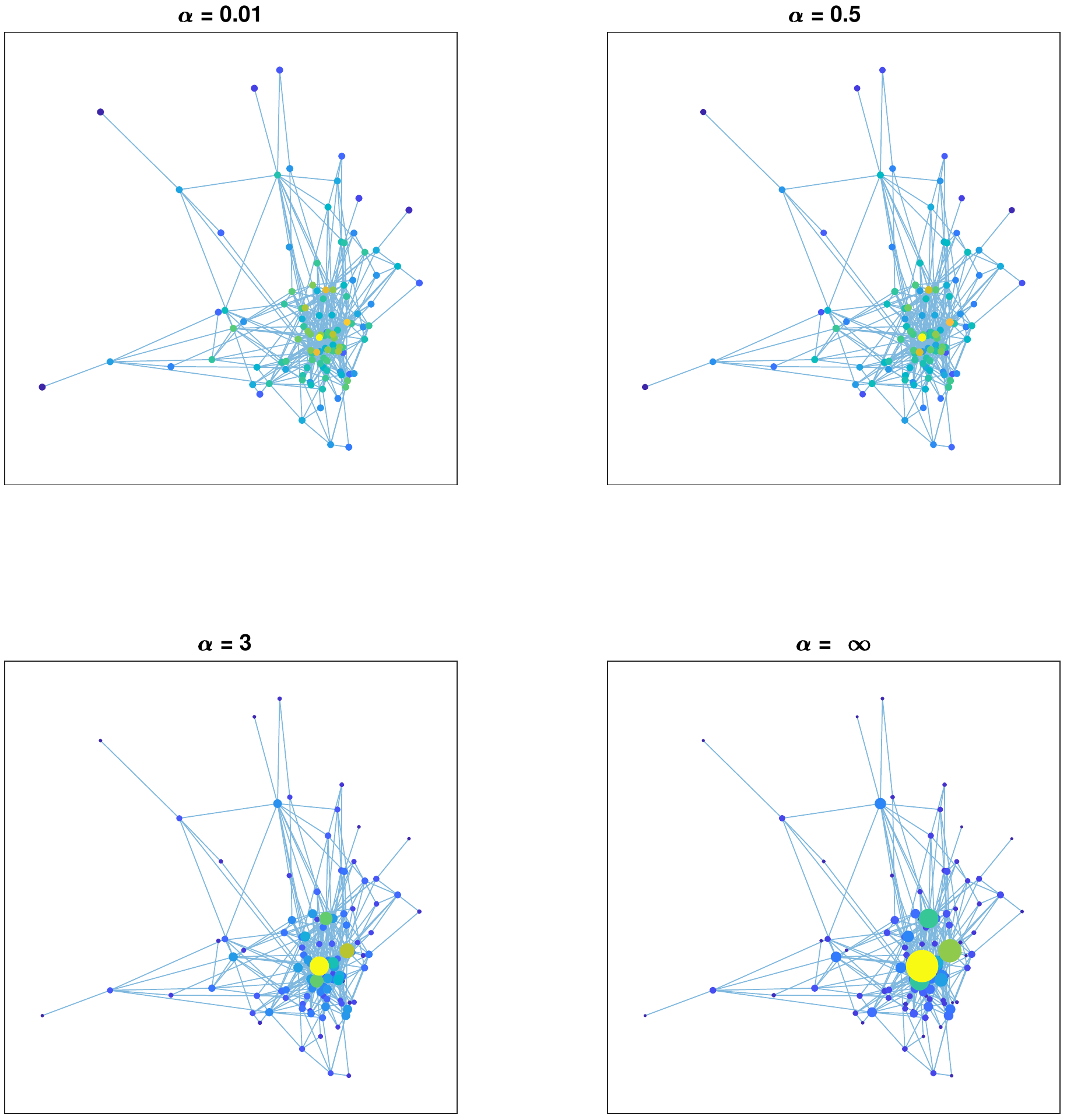}
	\caption{Localization of the PageRank values for the graph \texttt{Newman/adjnoun}  contains the network of common adjective and noun adjacencies for the novel ``David Copperfield'' by C. Dickens, as described in~\cite{MR2282139}. We have computed the PageRank vector for different values of $\alpha$ and plotted here the same graph with the nodes of the size representing the value assigned them by the algorithm. All values are consistently scaled. We observe that as $\alpha$ decreases the localization of the measure is sensibly reduced.}
	\label{fig:localization}
\end{figure}

The nonlocality properties of $\mathbf s_\alpha$ are useful in a number of situations. As an example of the improvements one can obtain by tuning the parameter $\alpha$, while fixing the choice of the distance to the shortest-path distance $\delta_\Gamma$, in the next section we consider the link prediction task with rooted PageRank similarity. %

\subsection{Link prediction}\label{sec:link_prediction}

Link prediction is an important task in network analysis, which consists of  the problem of predicting the existence of one or more missing (unobserved) edges in a given instance of a network $\Gamma$~\cite{doi:10.1002/asi.20591,MR3014051}. More precisely, we suppose having a snapshot $\Gamma_0 = (V,E_0)$ at time $t_0$ of a graph, and we want to guess what edges will be added at a subsequent time step $t_1$, in which the graph becomes $\Gamma_1 = (V,E_1)$, with $E_1 = E_0 \cup E_{\text{add}}$, and $E_{\text{add}} \subset V \times V \setminus E_0$. Two typical scenarios where this problem applies are the case of an evolving network and the case of  network data  affected by noise,  where it is suspected that a certain number of edges are missing.

A  successful approach for  link prediction  works by first assigning every edge $e_{i \rightarrow j}$ in $V \times V \setminus E_0$  a score, $\score(i,j)$, based on the graph $\Gamma_0$. In this way a ranked list of edges is produced, in decreasing order of $\score(i,j)$, and the new edges defining $E_1$ are taken as the edges with higher score.

Rooted (or seeded) PageRank is an established method for assigning such scores, based on the PageRank transition matrix. %
By extending that method, we consider a nonlocal rooted PageRank similarity, where the whole matrix of the scores of all the edges of type $e_{i \rightarrow j}$ for the parameter $\alpha$ is defined by
\begin{equation}\label{eq:similarity_for_link_prediction}
(S_\alpha)_{i,j} = \score_\alpha(i,j) = (X_\alpha + X_\alpha^T)_{i,j}, \qquad X_\alpha = (1-c)(I - c P_\alpha^T)^{-1},
\end{equation}
being $P_\alpha$ the nonlocal transition probability matrix in~\eqref{eq:levy_transition}.  Note that, due to Lemma~\ref{lem:smoothing_family}, we have  that $\score_\alpha$ coincides with the standard rooted PageRank similarity score when $\alpha\to\infty$. 

{In what follows we compare the link prediction performance of the nonlocal PageRank with the one obtained by using the standard PageRank. The test networks are both real-world examples: \texttt{adjnoun}, \texttt{USAir97}, \texttt{cage9}, \texttt{EuAirComplete}; and synthetic graphs \texttt{delaunay\_n10}, \texttt{delaunay\_n12}, and \texttt{3elt}. The network \texttt{EuAirComplete} is the flattened version of the multilayer network \texttt{EuAir} from~\cite{Cardillo2013}, representing the connection of airports in the European Union by the different airlines operating there. We subtract from each of the graphs $10\%$ of the edges chosen uniformly at random and we try to guess back all of them. To determine the parameters needed for the predictions, namely the $\alpha$ and $c$ parameters for the nonlocal Pagerank, and the $c$ parameter for the rooted PageRank, we apply a 10-fold cross-validation procedure on the smaller graph. The parameters are chosen on the grids $c\in\{0.1, 0.2, 0.3, 0.4, 0.5, 0.6, 0.7, 0.85, 0.9, 0.99\}$ and $\alpha \in \{0.1, 0.2, 0.3, 0.4, 0.5, 0.6, 0.7, 0.8, 0.9, 1, 2, 3, 5\}$. The whole procedure is then repeated fifteen times and the  results obtained are then reported in Figure~\ref{fig:link_prediction} as boxplots showing median and quartiles of the obtained prediction accuracy on the fifteen trials. What we observe is that harvesting information from far-away nodes, i.e., exploiting the slower decay of the similarity matrix $S_\alpha$ in~\eqref{eq:similarity_for_link_prediction}, yields better results. When such behavior is not observed the cross-validation pushes then the solution to be the one obtained with the rooted PageRank.}
For what concerns the cost, here it is dominated by the $O(n^3)$ matrix inversion in~\eqref{eq:similarity_for_link_prediction}, thus the overall cost of both the standard and the nonlocal PageRank are comparable. 
\begin{figure}[t]
	\centering
	\includegraphics[width=.8\textwidth]{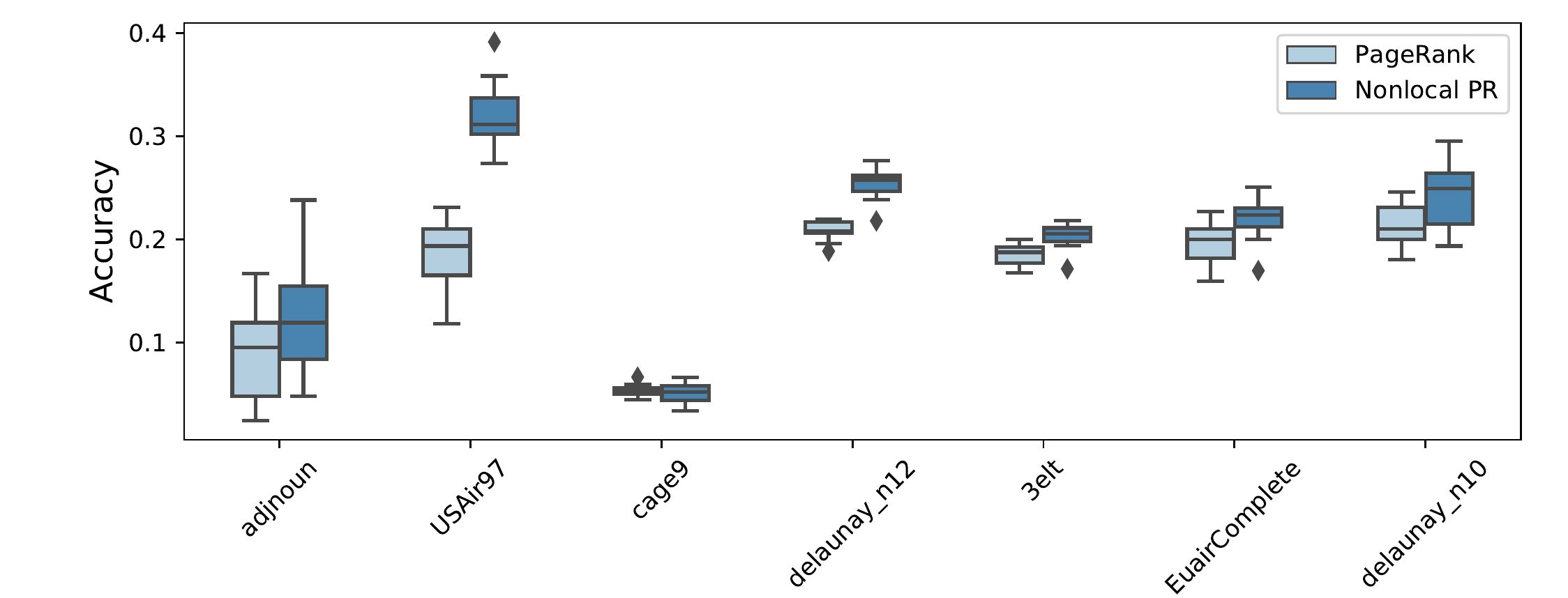}
	\caption{{Boxplot showing median and quartiles  of prediction accuracy over 15 trials on seven synthetic benchmarks and real-world datasets, for standard rooted PageRank versus nonlocal PageRank with smoothing function $f_\alpha(x) = x^{-\alpha}$. For each trial we remove  $10\%$ of edges chosen uniformly at random and we try to guess back all of them, choosing the parameters $\alpha$ and $c$ for nonlocal PageRank and the parameter $c$ for rooted PageRank via 10-fold cross validation on the smaller graph (with the removed edges) using the grids $c\in\{0.1, 0.2, 0.3, 0.4, 0.5, 0.6, 0.7, 0.85, 0.9, 0.99\}$ and $\alpha \in \{0.1, 0.2, 0.3, 0.4, 0.5, 0.6, 0.7, 0.8, 0.9, 1, 2, 3, 5\}$. }}
	\label{fig:link_prediction}
\end{figure}

\section{Choosing the appropriate distance}
\label{sec:choosing_the_distance}

The shortest path distance in Definition~\ref{def:shortestpath_distance} is not the only possible choice for generating the transition matrix in~\eqref{eq:levy_transition}. Being able to choose the distance function $\delta$ is an additional degree of flexibility of the proposed nonlocal PageRank which allows us to adapt the model more tightly to the problem. %
In principle any graph distance, respectively metric, can be adopted to define the transition probability matrix 
and the choice is a matter of modeling reasons. 

For illustration purposes, in the next section  we first describe an example of graph metric obtained by using the logarithmic distance. Then, in Section~\ref{sec:metro} we consider the London underground train multilayer network and define a problem-dependent \textit{metro distance}. This distance takes into account for the multiple layers and allows us to improve the centrality assignment of the train stations, when compared to independent station usage data we collected from \cite{london_usage}.

\subsection{An example of digraph metric: Logarithmic distance}

We illustrate here the behaviour of the nonlocal PageRank obtained with a distance that generates a metric on the graph, i.e., that is both symmetric, and satisfies the triangle inequality. To introduce such metric, called Logarithmic distance~\cite{MR2822195,MR2755905}, we need to start from a particular \emph{proximity} (\emph{similarity}) measure, $S = (s_{i,j}) = (s(i,j))$. Such measure is defined in terms of  the following  Laplacian matrix~$L$ of the digraph~$\Gamma$
\begin{equation}\label{eq:laplacian_matrix}
L = D_{\mathrm{out}} - A, 
\end{equation}
as
\begin{equation}\label{eq:similarity_matrix}S = (I + L)^{-1}\, ,
\end{equation}
where $D_{\mathrm{out}}$ is the diagonal matrix of the out degrees in~\eqref{eq:diagonalmatrix_of_out_degrees} and $A$ is the adjacency matrix of $\Gamma$.

This measure still accounts for the lack of symmetry of the underlying digraph, i.e., of its adjacency matrix, and satisfies both the \emph{transition inequality} and \emph{graph bottleneck identity}, i.e., it is such that $s_{i,j} s_{j,k} \leq s_{i,k}s_{j,j}$, and $s_{i,j}s_{j,k} = s_{i,k}s_{j,j}$ if and only if every path from $i$ to $k$ contains $j$. In the case of undirected graphs \eqref{eq:similarity_matrix} is usually called the regularized Laplacian kernel. Such definition is indeed well posed, i.e., we are ensured that the matrix $I+L$ can be inverted,  because $I+L$ is an example of a nonsingular $M$-matrix (cf.~\cite{MR1298430}). Note that this is sufficient to guarantee also that all the elements of $S$ are nonnegative; see~\cite{MR1298430} for further details.

To obtain the Logarithmic distance based on the similarity $S$ we then build the matrix $H = (h_{i,j})$ and the vector $\mathbf{h}$ defined as 
\begin{equation}\label{eq:logarithmic_map}
h_{i,j} = \log(s_{i,j}) \leq 0, \quad \mathbf{h} = (h_{1,1},h_{2,2},\ldots,h_{n,n})^T
\end{equation}
from which we define the \emph{logarithmic distance} as
\begin{equation}\label{eq:logarithmic_distance}
\delta_{\mathrm{log}}(i,j) =  \left(\frac{1}{2}(U + U^T)\right)_{i,j}, \qquad U = \mathbf{h}\boldsymbol{1}^T - H.
\end{equation}
We stress that $\delta_{\mathrm{log}}$, differently from the shortest--path distance in~Definition~\ref{def:shortestpath_distance}, generates a metric on the digraph,  see e.g.\ \cite[Theorem~1]{MR2822195}. This can also be observed by the example network in Figure~\ref{fig:distance_comparison} in which we show the comparison of the entries of the distance matrix for the shortest path distance (central panel) and for the logarithmic distance (right panel) as compared to the pattern of the adjacency matrix of the real world network {\tt gre\_115} (left panel).
\begin{figure}[t]
	\centering
	\includegraphics[width=0.9\columnwidth]{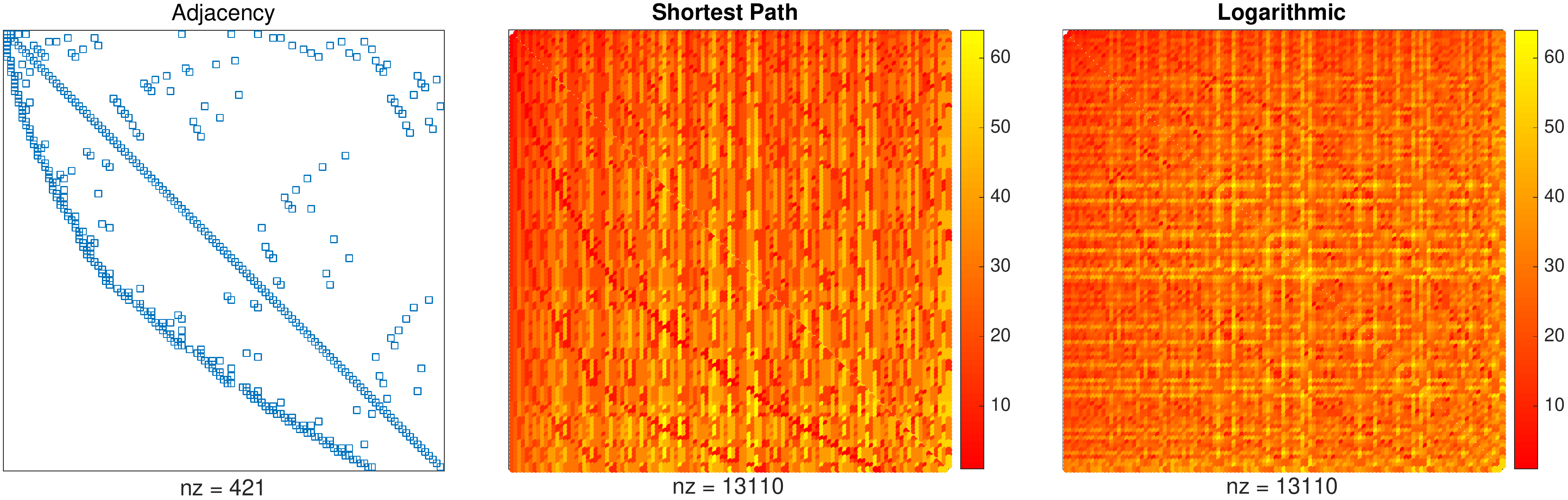}
	\caption{\texttt{HB/gre\_115} digraph. The figure shows the adjacency and distance matrices for the loop--less version of the \texttt{HB/gre\_115} digraph.}
	\label{fig:distance_comparison}
\end{figure}

Instead, in order to better grasp the differences between the centralities obtained, in Figure  \ref{fig:expvslog}  we scatter plot the behavior of the nonlocal PageRank on  the dataset {\tt USAir97}, for the two different smoothing functions $f(x)_\alpha=1/x^\alpha$ and $f_\alpha(x) = e^{-\alpha x}$ and the two choices of graph distances $\delta=\delta_\Gamma$ and $\delta=\delta_{\mathrm{log}}$.

\begin{figure}[t]
	\centering
	\includegraphics[width=0.9\columnwidth]{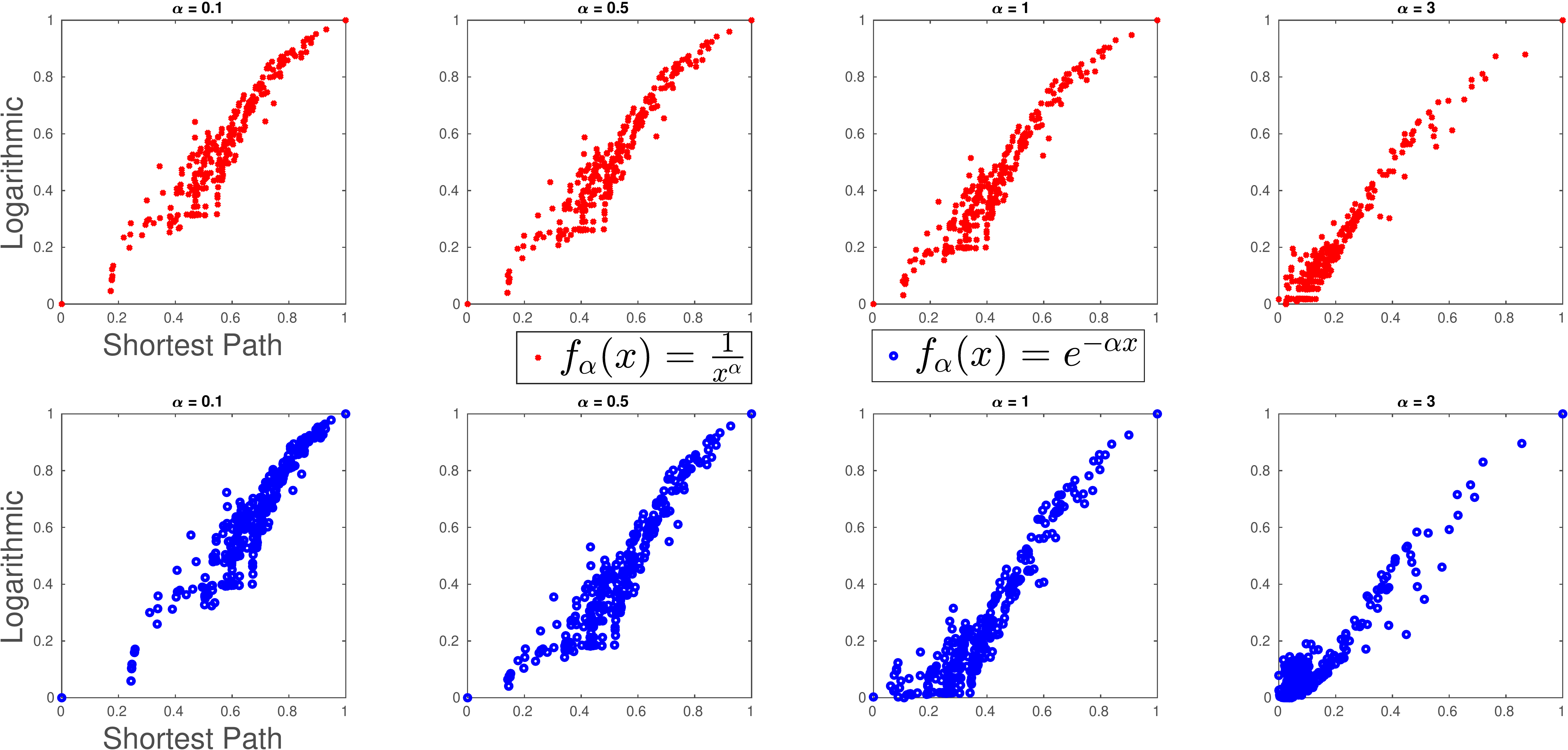}
	\caption{\texttt{USAir97} digraph. The figure shows scatter plots of Nonlocal PageRank obtained using shortest path distance  vs Nonlocal PageRank obtained using logarithmic distance for different  $f_\alpha(x)$ and different values of $\alpha$.}
	\label{fig:expvslog}
\end{figure}

\subsection{Metro distance and the London's  train multilayer network}\label{sec:metro}
In this final example scenario of applicability for the proposed nonlocal PageRank model we consider the ranking in term of importance of urban railways. %
Tools from network science have proved to be valuable in the study of urban transport \cite{derrible2012network,To2015} and here we consider the use of nonlocal PageRank network centrality in the case of
the London train network. Nodes in the network represent stations, and we seek a
distance $\delta$ so that the resulting centrality measure correlates well with passenger usage. Such a measure,
which requires only the topological connectivity structure, offers helpful information
at the design stage. More importantly, it can be used in what-if-scenario testing, in
order to predict the effect of changes, including unplanned network disruptions. 

The London train network is an undirected transportation network representing connections between tube train stations of the city of London. The dataset {\tt tube} we consider here is a multilayer version of the train network, where each underground line corresponds to one layer. The dataset has been generated using the data from \cite{de2014navigability} as baseline.  The aggregate network consists of one connected component with 271 nodes and 315 edges with nonzero weights. Each edge encodes the information about the membership of a  given node to one---or more than one in the case of  intersections---of the $k=11$ different underground lines.
We collect additional passenger data from \cite{london_usage}: for each train station of the above
network, we collect the number of passengers entering or exiting that station per
year. We collect data for ten years: from 2008 to 2017. Our data is publicly available
at \url{https://github.com/Cirdans-Home/NonLocalPageRank}.

Motivated by the question of whether we can identify highly populated stations by exploiting only the topology of connections between stations,  in this section we study the behavior of the nonlocal PageRank centrality with different choices of the distance function $\delta$ and of the decaying parameter $\alpha$. In particular, in the next subsection we will design a distance that is specifically conceived for this issue and we will show that this indeed helps boosting the performance of the centrality model in this context.

\subsubsection*{The \textit{Metro Distance}}
Every day experience using underground trains suggests that not all the paths between two connected nodes $i$ and $j$ of the network are equally attractive: we posit that users tend to prefer paths that avoid line changes (or minimize them) even if that means choosing a path that is longer in terms of the number of stations that the path involves.

This argument suggests that the shortest path distance $\delta_\Gamma(i,j)$ is not appropriate to faithfully model  the graph exploration of a rail traveler. Hence, we consider here a natural  modification of the shortest path distance $\delta_\Gamma$, which we will call \textit{metro distance} and which we define as follows 
\begin{equation}\label{eq:metro_distance}
\delta_M(i,j) = \begin{cases}
\delta_\Gamma(i,j) & \hbox{ if } i,j \hbox{ are on \textit{the same metro line}}, \\
\delta_\Gamma(i,j)+ C(i,j) & \hbox{ if } i,j \hbox{ are on \textit{different metro lines}}, 
\end{cases}\,
\end{equation}
where $C(i,j)$ is the number of times a traveler needs to change  train line when traveling from node $i$ to node $j$.
In other words, the metro distance penalizes  paths using as penalization parameter the number of times the passenger changes line. 

Our aim is now to test the extent to which the nonlocal PageRank centrality with $\delta_M$ can identify nodes that perform better, in terms of total passenger usage, than other PageRanks.
In order to measure  this type of  performance for the different ranking strategies, we consider the \emph{Intersection Similarity} (ISIM), whose definition we briefly recall. The intersection similarity is a measure that compares two ranked lists that may not contain the same elements.  It is defined as follows~\cite{fagin2003comparing}: let $\b p, \b q$ be two  ranked lists of $n$ elements we want to compare. %
The intersection similarity of $\b p$ and $\b q$ is the vector $\text{ISIM}(\b p,\b q)$ with entries 
\begin{equation*}\label{eq:isim}
\text{ISIM}(\b p,\b q)_k = \frac{1}{k}\sum_{j=1}^k\frac{\left|\Delta\left((p_1,\dots, p_j),(q_1,\dots,q_j)\right)\right|}{2j}, \qquad k=1,\dots, n
\end{equation*}
where, for sets $S,T$, $|S|$ denotes the cardinality and $\Delta$ is the symmetric difference operator $\Delta(S,T) := (S\setminus T)\cup (T\setminus S)$.

When the first $k$ entries in $\b p$ and $\b q$ are completely different ISIM$(\b p,\b q)_k$ is equal to 1, whereas ISIM$(\b p,\b q)_k=0$ 
if and only if the first $k$ entries in $\b p$ and $\b q$ coincide exactly.
More in general, lower values in  ISIM$(\b p,\b q)$ imply a better matching between $\b p$ and $\b q$. 

In Figure~\ref{fig:isim_comparison} we aim at comparing the top fifteen stations identified by the following three models 
\begin{itemize}
	\item the nonlocal PageRank with $\delta_\Gamma$ (``\textit{SP distance}''),
	\item the nonlocal PageRank with $\delta_M$ (``\textit{Metro distance}'') ,
	\item the standard PageRank. 
\end{itemize}
with the ``ground truth'', i.e., the actual fifteen stations with the
highest number of passengers, whose corresponding ranked list we denote by $\rho$. In particular, we compute the intersection similarity between the ranking of any of the three above PageRanks and $\rho$, which we denote respectively by ISIM$(\delta_\Gamma)$, ISIM$(\delta_M)$ and ISIM$(pr)$.

The two curves in the left panel of Figure~\ref{fig:isim_comparison} show the ratios 
$$
\frac{\text{ISIM}(\delta_\Gamma)_{15}}{\text{ISIM}(pr)_{15}} \qquad \text{and} \qquad \frac{\text{ISIM}(\delta_M)_{15}}{\text{ISIM}(pr)_{15}}
$$
with blue and red lines, respectively, 
for different values of the parameter $\alpha$, computed with the power-law decaying function $f_\alpha(x)=1/x^\alpha$. 

Whereas, the central panel compares the cumulative sum
of passenger usage values for the 15 top ranked stations for the three rankings. This comparison is done by showing both the ratio between the cumulative sum of passenger usage values corresponding to $\delta_\Gamma$ and the one corresponding to the standard PageRank (blue line) and the ratio between the one corresponding to $\delta_M$ and the one corresponding to the standard PageRank (red line), for different values of  $\alpha$. 
In both panels we observe the metro distance $\delta_M$ performing generally better than the standard shortest path distance, with relative optimal performance obtained for $\alpha=1.7$. 
This is further highlighted by the panel on the right of Figure \ref{fig:isim_comparison}, where we compare the first 25 entries of the three different intersection similarity vectors ISIM$(\delta_\Gamma)$, ISIM$(\delta_M)$ and ISIM$(pr)$, for $\alpha=1.7$.

\begin{figure}[htbp]
	\centering
	\includegraphics[width=\columnwidth]{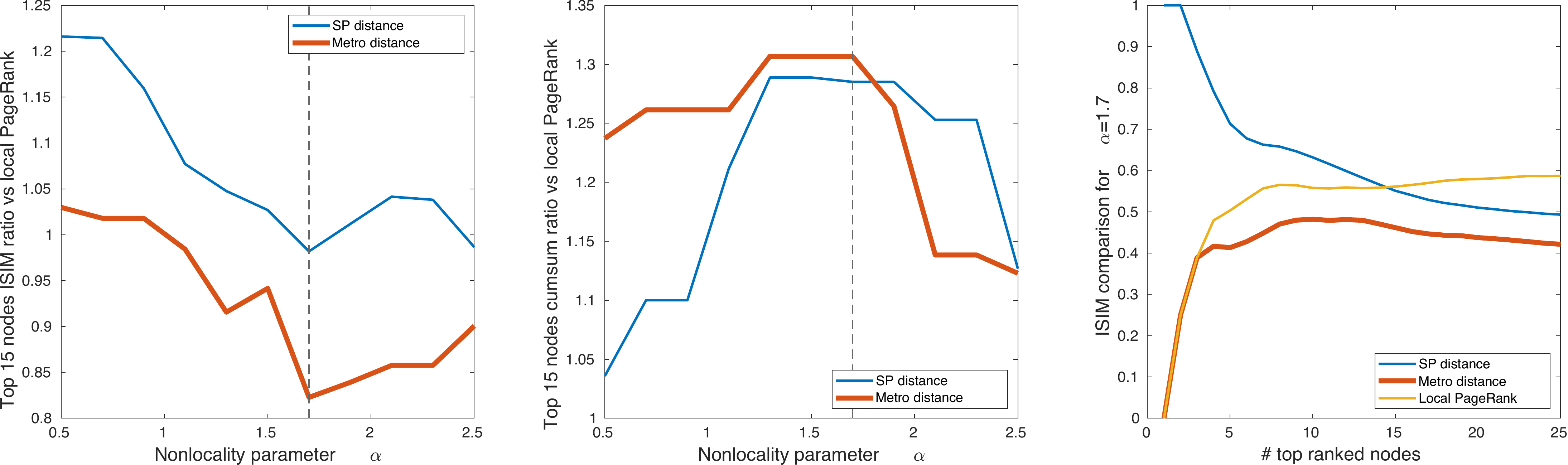}
	\caption{Non-local PageRank ranking performance. Year 2017}
	\label{fig:isim_comparison}
\end{figure}

\begin{table}[ht]
	\begin{center}
		\begin{adjustbox}{width=.9\textwidth}
			\sf \small
			\begin{tabular}{|c|c|c|c|c|c|c|c|c|c|c|c|}
				\hline
				& \textbf{$k$} & \textbf{2017} & \textbf{2016} & \textbf{2015} & \textbf{2014} & \textbf{2013} & \textbf{2012} & \textbf{2011} & \textbf{2010} & \textbf{2009} & \textbf{2008}\\
				\hline \hline 
				\multirow{4}{*}{\textbf{Ground Truth}} & 5 & 421.6895 & 433.7675 & 437.0499 & 443.6953 & 413.9785 & 399.7917 & 386.0158 & 366.9199 & 353.7508 & 359.7531 \\ 
				& 15 & 891.2556 & 918.167 & 920.4402 & 918.3901 & 867.6314 & 835.8205 & 805.7167 & 755.9158 & 731.6231 & 739.0077 \\ 
				& 45 & 1614.8047 & 1665.0602 & 1615.3159 & 1615.757 & 1542.1181 & 1490.7425 & 1430.1159 & 1374.3307 & 1329.411 & 1348.2863 \\ 
				
				\hline 
				\multirow{4}{*}{\textbf{``Local'' PageRank}} & 5 & 286.787 & 294.6265 & 288.0692 & 286.9402 & 273.3404 & 264.2259 & 255.6953 & 245.6262 & 229.4468 & 230.0058 \\
				& 15 & 580.5896 & 592.7364 & 588.1613 & 589.1657 & 545.3233 & 522.9533 & 504.7835 & 480.4868 & 462.1433 & 473.7567 \\
				& 45 & 1260.4716 & 1294.76 & 1259.7501 & 1259.2124 & 1204.2262 & 1159.854 & 1116.8103 & 1070.5588 & 1036.5845 & 1050.4783 \\

				\hline
				\multirow{4}{*}{\textbf{Nonlocal PageRank  $\boldsymbol{\delta_\Gamma}$}} & 5 & 341.4023 & 349.6481 & 349.839 & 353.9939 & 323.0698 & 311.4126 & 297.0677 & 279.9141 & 268.7073 & 271.4097 \\ 
				& 15 & 746.1333 & 768.6693 & 762.5181 & 751.9844 & 723.7442 & 699.2322 & 675.9891 & 648.6213 & 624.0347 & 629.1298 \\ 
				& 45 & 1302.5842 & 1320.4892 & 1320.0771 & 1324.7913 & 1269.6994 & 1222.4786 & 1181.0192 & 1150.9596 & 1123.1939 & 1133.568 \\
				
				\hline
				\multirow{4}{*}{\textbf{Nonlocal PageRank $\boldsymbol{\delta_M}$}} & 5 & 341.4023 & 349.6481 & 349.839 & 353.9939 & 323.0698 & 311.4126 & 297.0677 & 279.9141 & 268.7073 & 271.4097 \\ 
				& 15 & 758.6355 & 785.5156 & 774.1294 & 761.5849 & 733.9973 & 709.7501 & 686.8022 & 659.9881 & 634.5561 & 639.799 \\ 
				& 45 & 1352.7352 & 1373.8639 & 1362.7264 & 1366.1506 & 1300.7856 & 1251.9737 & 1203.3251 & 1154.5498 & 1124.5568 & 1139.3027 \\
				
				\hline
			\end{tabular}
		\end{adjustbox}
	\end{center}
	\vspace{1em}
	\caption{Millions of passengers per year using the top $k=5,15,45$ stations ranked according to the standard ``\textit{local PageRank}'' the nonlocal PageRank with $\delta_\Gamma$ ($\alpha=1.7$) and the nonlocal PageRank with $\delta_M$ ($\alpha=1.7$). We consider the range 2017--2008. }\label{tab:million_passegers}
\end{table}

\begin{table}[htpb]
	\begin{center}
		\begin{adjustbox}{width=.9\textwidth}
			\sf \small
			\begin{tabular}{llllllll}
				\textbf{Ground Truth} &  & \textbf{``Local'' PageRank} &  & \textbf{Nonlocal PageRank  $\boldsymbol{\delta_\Gamma}$} &  & \textbf{Nonlocal PageRank $\boldsymbol{\delta_M}$} &  \\ 
				\hline 
				King's Cross & 97.9183 & King's Cross & 97.9183 & Green Park & 39.3382 & King's Cross & 97.9183 \\ 
				Waterloo & 91.2706 & Baker Str. & 28.7846 & Baker Str. & 28.7846 & Baker Str. & 28.7846 \\ 
				Oxford Circus & 84.0906 & Paddington & 48.8225 & Oxford Circus & 84.0906 & Green Park & 39.3382 \\ 
				Victoria & 79.3593 & Earl's Court & 19.991 & King's Cross & 97.9183 & Oxford Circus & 84.0906 \\ 
				London Bridge & 69.0507 & Waterloo & 91.2706 & Waterloo & 91.2706 & Waterloo & 91.2706 \\ 
				Liverpool Str. & 67.7402 & Turnham Green & 6.1552 & Bond Str. & 38.8027 & Paddington & 48.8225 \\ 
				Stratford & 61.9904 & Green Park & 39.3382 & Bank & 30.8981 & Bank & 30.8981 \\ 
				Canary Wharf & 50.9136 & Oxford Circus & 84.0906 & Westminster & 25.5954 & Bond Str. & 38.8027 \\ 
				Paddington & 48.8225 & Stockwell & 11.6971 & Paddington & 48.8225 & Earl's Court & 19.991 \\ 
				Euston & 43.0737 & Liverpool Str. & 67.7402 & Liverpool Str. & 67.7402 & Euston & 43.0737 \\ 
				\hline 
			\end{tabular}
		\end{adjustbox}
	\end{center}
	\vspace{1em}
	\caption{Ten London train stations with highest ranking value, according to the standard ``\textit{local PageRank}'',  the nonlocal PageRank with $\delta_\Gamma$ and the nonlocal PageRank with $\delta_M$ ($\alpha=1.7$ in both cases). Year  2017. }\label{tab:ten_stations}
\end{table}

As shown in Figure~\ref{fig:isim_comparison}, 
a proper choice of the parameter $\alpha$,
allows for a more accurate ranking than the local PageRank. In particular, choosing $\alpha=1.7$ and the metro distance $\delta_M$, we are able to match the ground truth ranking more closely if compared with the local PageRank and the nonlocal PageRank employing the distance $\delta_\Gamma$. Finally, as the right panel of Figure~\ref{fig:isim_comparison} shows,
even thought  the ISIM performance with respect to the ground truth ranking coincides on the top ranked nodes for the standard PageRank and the nonlocal PageRank with \textit{Metro distance}, the nonlocal model allows us to obtain sensibly better performance in terms of ISIM when a greater number of ranked nodes is taken into account.  This issue is further corroborated from the results presented in Table~\ref{tab:million_passegers} where a similar behavior is observed across the ten--year span usage data;  here we compare the number of
passengers using the top $k = 5, 15, 45$ stations identified by local/nonlocal Pagerank as before. Finally, in Table~\ref{tab:ten_stations} and Figure~\ref{fig:london_groundtruth} we show the name and the geographic collocations of the top $10$ ranked stations according to the considered different ranking algorithms.
The overall emerging experimental evidence from Figures~\ref{fig:isim_comparison},~\ref{fig:london_groundtruth} and Tables~\ref{tab:million_passegers},~\ref{tab:ten_stations}, highlights how
the flexibility of the proposed model allows us to design centrality that better adapt to the specific problem. %
While the nonlocal PageRank model with the metro distance does not yield a perfect matching with the ground truth, it outperforms other PageRank models and obtains remarkable performance which we find particularly interesting given that the model  exploits only the topological structure of nodes and edges.  
\begin{figure}[htpb]
	\centering
	\includegraphics[width=0.45\columnwidth,clip,trim=2cm 3cm 1.5cm 3cm]{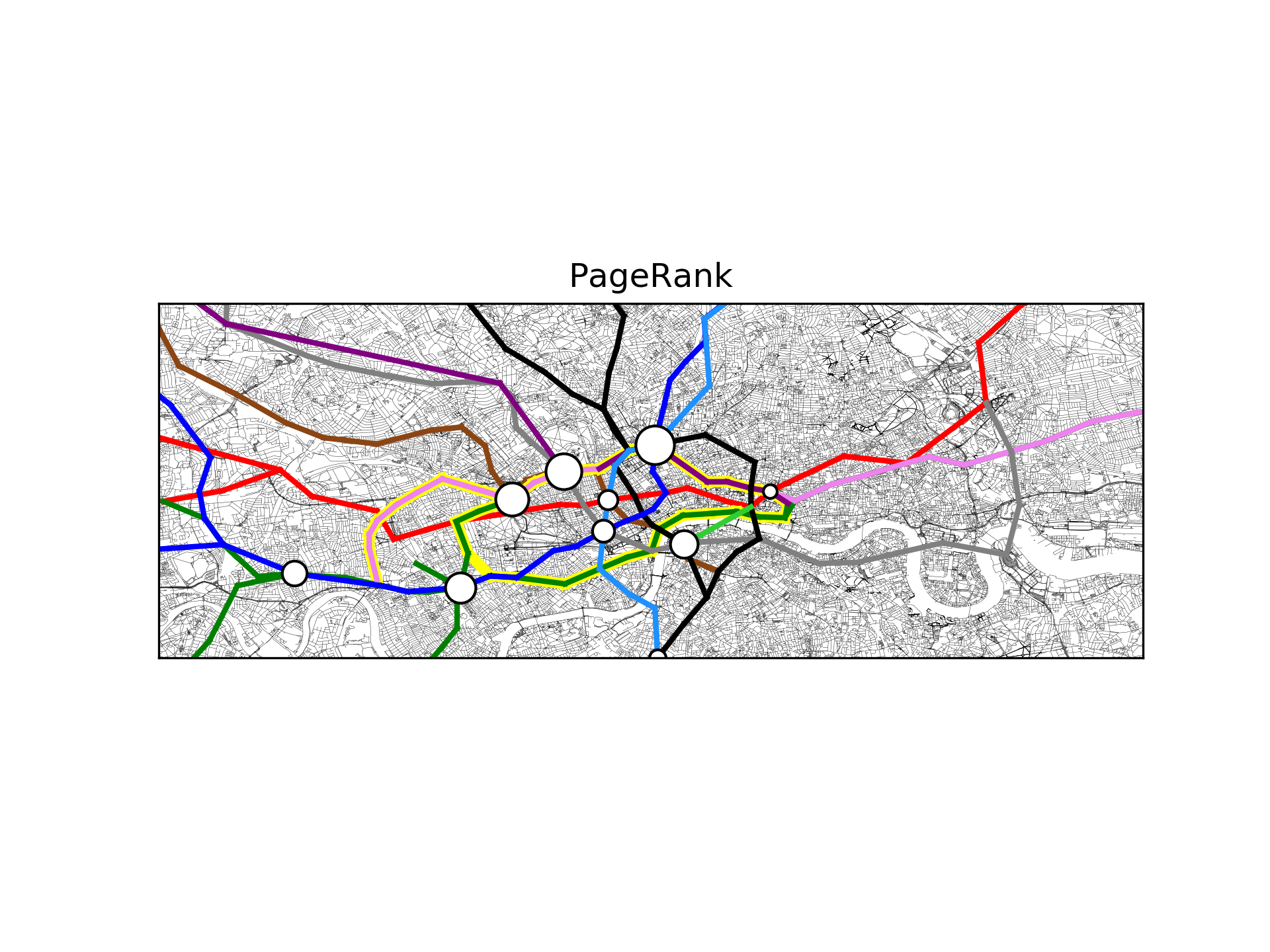}\hfill 
	\includegraphics[width=0.45\columnwidth,clip,trim=2cm 3cm 1.5cm 3cm]{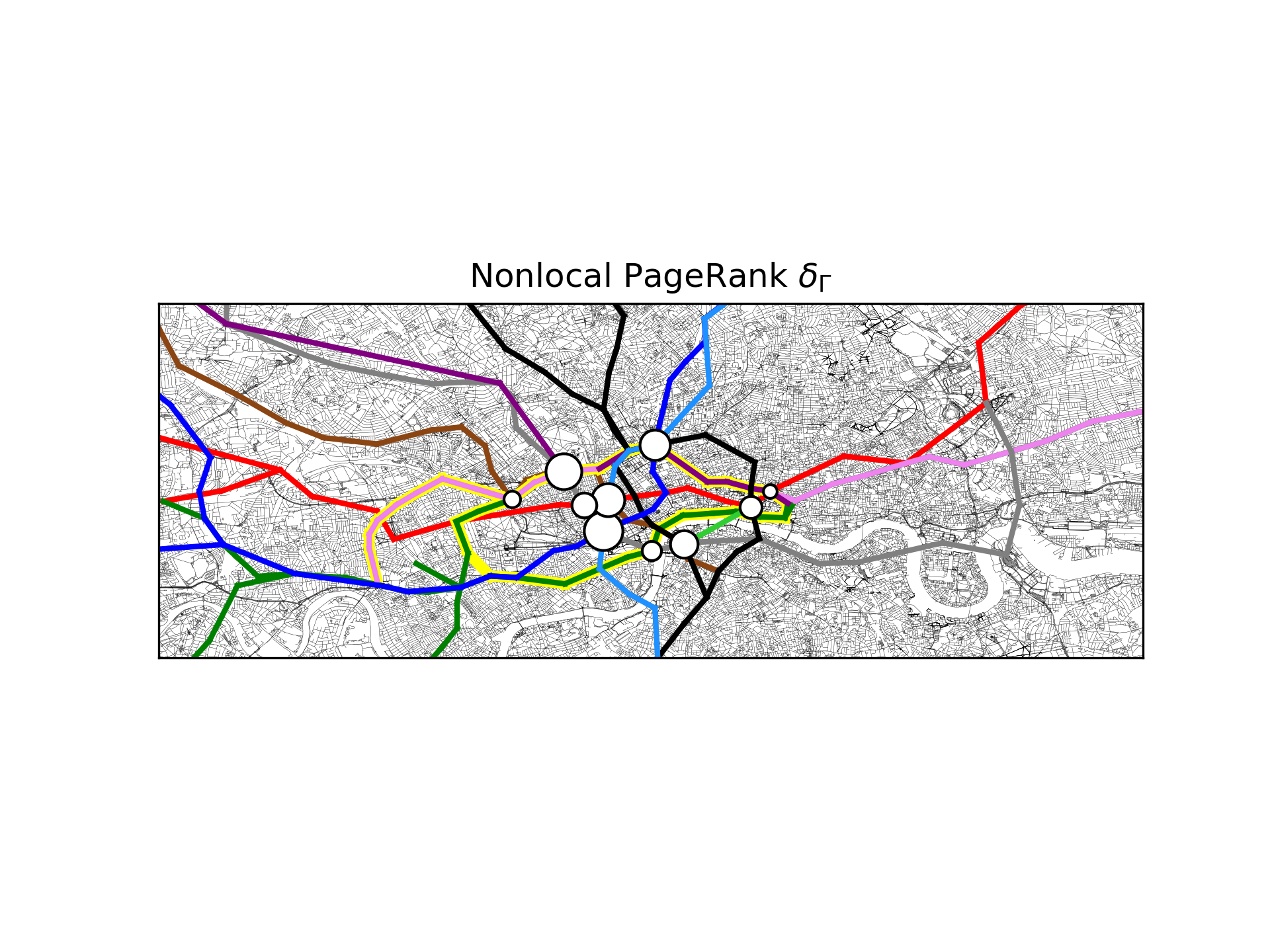}\\
	\includegraphics[width=0.45\columnwidth,clip,trim=2cm 3cm 1.5cm 3cm]{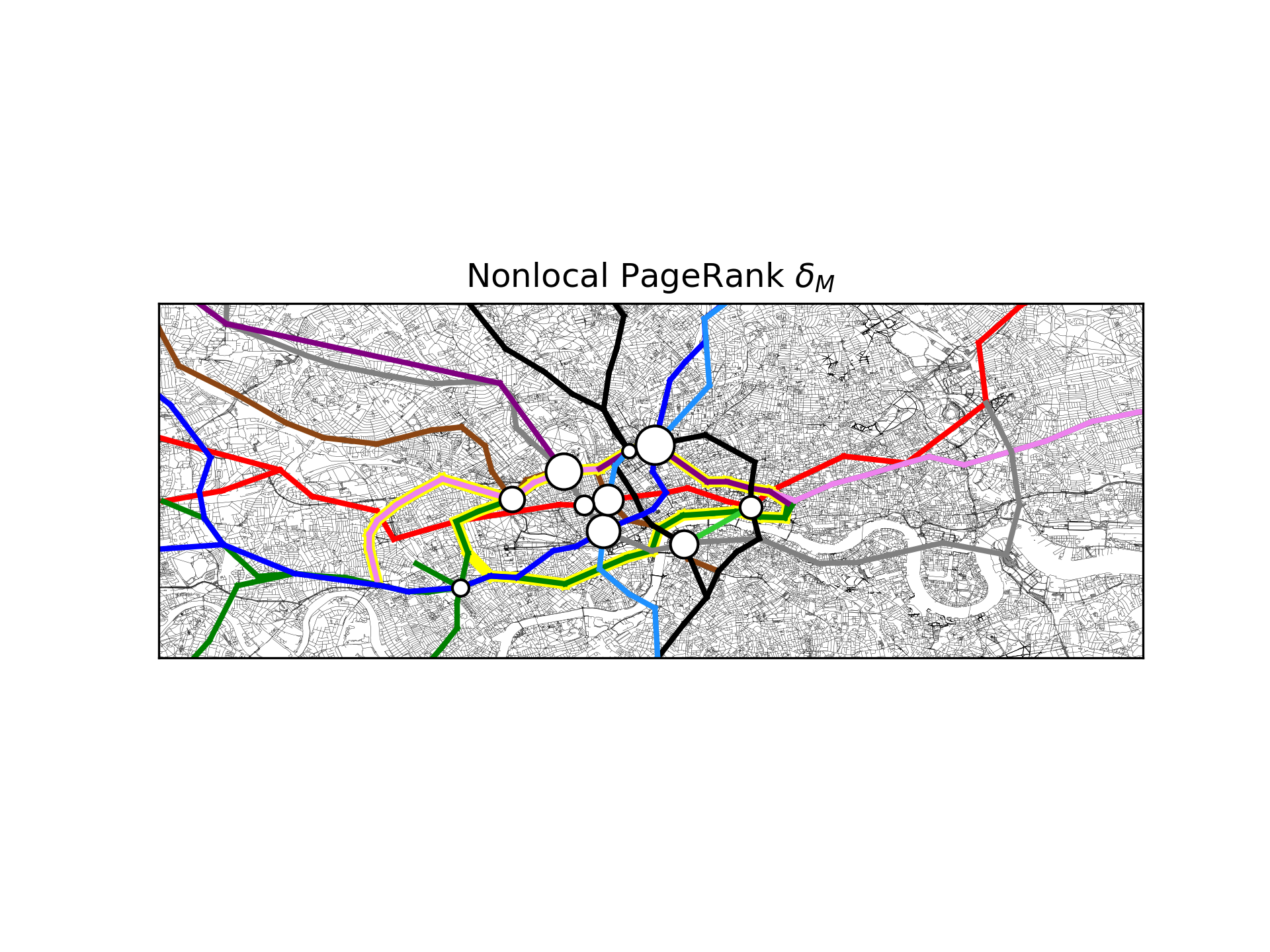}\hfill
	\includegraphics[width=0.45\columnwidth,clip,trim=2cm 3cm 1.5cm 3cm]{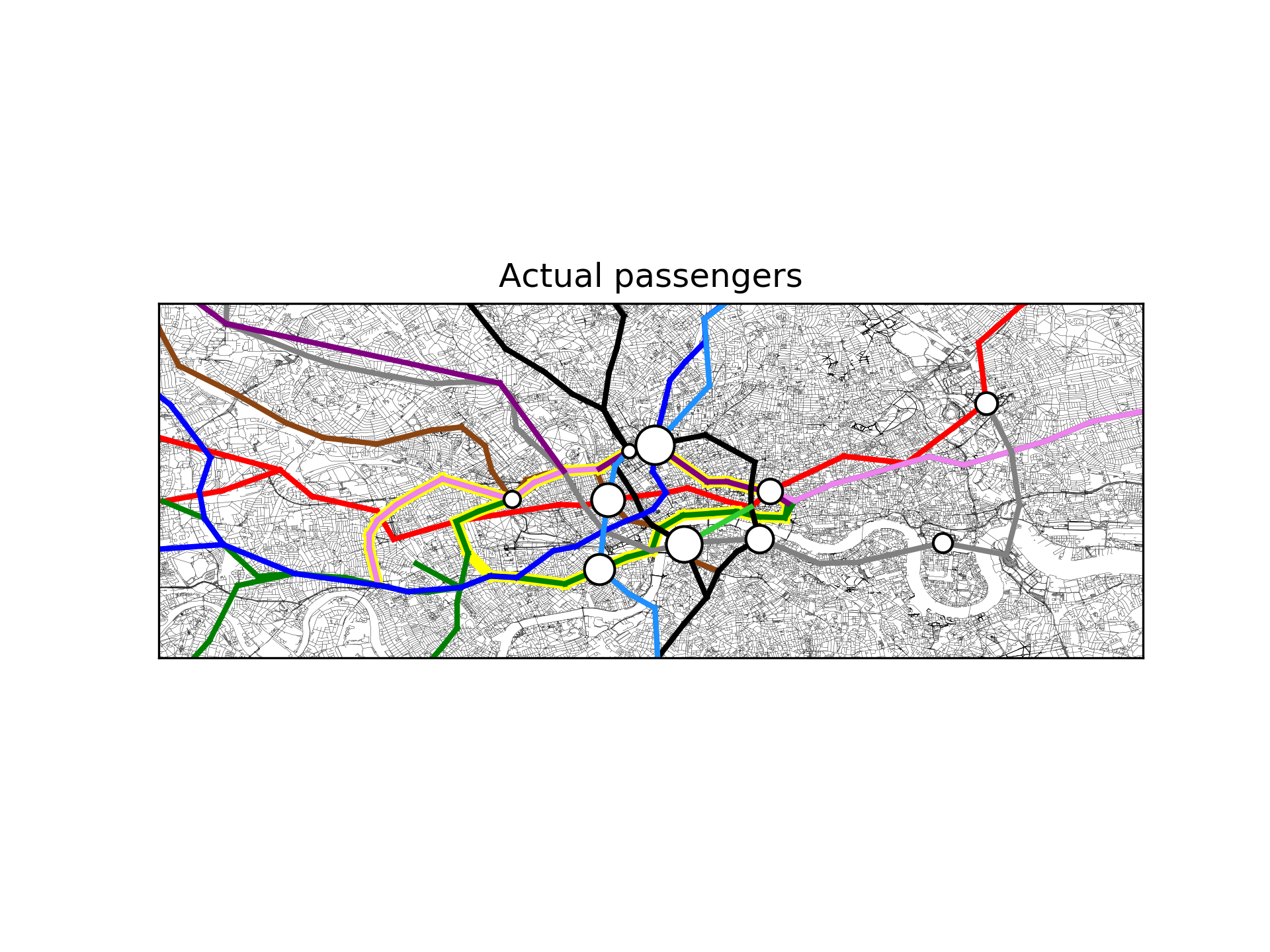}
	\caption{This figure shows the name and the geographic collocations of the top $10$ ranked stations according to the considered different ranking algorithms together with the top $10$ ranked stations with respect to the actual number of passengers.}
	\label{fig:london_groundtruth}
\end{figure}

\subsubsection*{Computation of the metro distance}
To obtain the penalized version of the shortest path distance  $\delta_M(i,j)$, see~\eqref{eq:metro_distance}, we synthesized from a \textit{multilayer} interpretation for the graph $\Gamma$. 
This is a particularly useful approach and, in this case, enabled us to encode the information coming from the connectedness of two nodes $i$ and $j$ via the line $k$. The multilayer structure of the graph can be naturally represented using a tensor $\mathcal{T} \in \mathbb{R}^{n \times n \times k}$ such that
\begin{equation}
\mathcal{T}_{ij\ell}= \begin{cases}
1 \hbox{ if } i \sim j \hbox{ on the line } \ell\\
0 \hbox{ otherwise.}
\end{cases}
\end{equation}
In the following we will use the following Matlab notation: $\mathcal{T}(:,:,\ell):=(\mathcal{T}(i,j,\ell))_{i,j \in \{1,\dots,n\}}$.
Observe that for every $\ell=1,\dots,k$, since we are considering undirected connections among nodes, we have $\mathcal{T}(:,:,\ell)=\mathcal{T}(:,:,\ell)^T$; moreover, $\operatorname{sign}(\sum_{\ell=1}^{k}T(:,:,\ell))$ returns exactly the adjacency matrix of the full graph $\Gamma$.
By mapping each node $i \in \{1,\dots,n\}$ of the graph in $i_1,\dots,i_k$, it is possible to form the block matrix 
\begin{equation}
T=\begin{bmatrix}
T_{1,1} & \dots & T_{1,k} \\
\vdots & \vdots & \vdots \\
T_{k,1} & \dots & T_{k,k} \\
\end{bmatrix} \in \mathbb{R}^{nk\times nk}
\end{equation}
such that $T_{\ell,\ell}=\mathcal{T(:,:,\ell)}$ and $(T_{\ell_1,\ell_2})_{i,i}=(T_{\ell_2,\ell_1})_{i,i}=1$ if the node $i \in \{1,\dots,n\}$ is at the intersection of the metro lines $\ell_1$ and $\ell_2$; the remaining elements of $T_{\ell_1,\ell_2}$ and $T_{\ell_2,\ell_1}$ are set to zero. 
Now, in the \textit{expanded} graph $\overline{\Gamma}=adj(T)$ with nodes $\{1_1,\dots,n_1,\dots, 1_k, \dots, n_k\}$, let us analyze the path  $j_{\ell_1}\rightarrow i_{\ell_2} \rightarrow u_{\ell_3}\rightarrow e_{\ell_4} $ with $\{j,i,u,e\} \in \{1,\dots,n\}$; it is easy to recognize that the case $i=u$ corresponds to the case where a traveler changes metro line ${\ell_2}$ into metro line ${\ell_3}$ at the node $i$. The metro distance $\delta_M(i,j)$ is then obtained considering
$\delta_M(i,j)=\min_{\ell_1=1,\dots,k,\ell_2=1,\dots,k}\delta_{\Gamma}(i_{\ell_1},j_{\ell_2})$, being $\delta_\Gamma$ the shortest path distance of the nodes $i_{\ell_1},j_{\ell_2}$ computed in~$\overline{\Gamma}$.
As in most of the other examples, the main computational cost is represented by the computation of the shortest path distance on $\overline{\Gamma}$. 

\section{Conclusions and Future Work}

In this work we have introduced a nonlocal version of the classic PageRank model. The new version encompasses a nonlocal navigation strategy of the underlying network, permitting the usage of any suitable graph distance. Generalizing the L\'evy and exponential transition models, we have introduced a general definition for a class of functions which can be used to modulate the range of the interactions.

With the approach presented here it is possible to mitigate several typical phenomena occurring in eigenvector centralities, such as the phenomenon of localization of the measure, and the issue of the assignation of numerically indistinguishable values to nodes that are in the lower part of the ranking. The mitigation of these behaviors increases the predictive power of the nonlocal PageRank when compared with the local PageRank, as it has been confirmed by the real-world applications presented in this work.

Even if the model we present encodes a completely different dynamics for the network interactions, if compared with the local approach, it is still possible to use a standard series of numerical tools for the efficient computations of PageRank vectors, see for example \cite{cipolla2017euler,brezinski2006pagerank,del2005fast,golub2006arnoldi}. The main differences are represented by the setup phase of the algorithm, i.e., by the need of computing the distance matrix for the underlying network, and by the fact that nonlocal transition matrices are in general not sparse. On the other hand, suitable choices of the smoothing functions $f_\alpha$ may lead to structured transition matrices (as in the case of fractional derivatives \cite{massei2019fast} e.g.) and exploring this line of research may lead to efficient methods for using nonlocal PageRank on large scale problems, an issue that will be object of future investigations.%
\bibliographystyle{abbrv}
\bibliography{bibliography}
\end{document}